%% file: main.tex
\documentclass[runningheads]{llncs}
\input{defs}

\input{commands}

\usepackage[paperheight=235mm,paperwidth=155mm,textwidth=12.2cm,textheight=19.3cm]{geometry}
\makeatletter
\def\@citecolor{blue}%
\def\@urlcolor{blue}%
\def\@linkcolor{blue}%

\def\orcidID#1{\smash{\href{http://orcid.org/#1}{\protect\raisebox{-1.25pt}{\protect\includegraphics{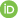}}}}}
\makeatother

\newcommand{\doublec}[1]{\ensuremath{[\![ #1 ]\!]}}

\newcommand{\arrowstar}[1]{\ensuremath{\xrightarrow[]{*}_{#1}}}

\usepackage[disable]{todonotes}
\newcommand{\chana}[1]{\todo[color=green!30]{\small #1}}
\newcommand{\chanain}[1]{\todo[color=green!30,inline]{#1}}
\newcommand{\bala}[1]{\todo[color=blue!30]{\small #1}}

\begin{document}
\title{Parameterized Analysis of Reconfigurable Broadcast Networks\thanks{This project has received funding from the European Research Council (ERC) under the European Union's Horizon 2020 research and innovation programme under grant agreement No 787367 (PaVeS).}}
\titlerunning{Parameterized Analysis of RBN}
%
\author{A. R. Balasubramanian\inst{1}\orcidID{0000-0002-7258-5445} \and
Lucie Guillou\inst{2}\orcidID{0000-0002-6101-2895} \and
Chana Weil-Kennedy\inst{3}(\Envelope)\orcidID{0000-0002-1351-8824}}
\authorrunning{A. R. Balasubramanian, L. Guillou, C. Weil-Kennedy}
\institute{Technical University of Munich
\email{bala.ayikudi@tum.de} \and 
ENS Rennes
\email{lucie.guillou@ens-rennes.fr} \and
Technical University of Munich
\email{chana.weilkennedy@in.tum.de} 
}%
\maketitle              
\begin{abstract}

Reconfigurable broadcast networks (RBN) are a model of distributed computation in which agents can broadcast messages to other agents using some underlying communication topology which
can change arbitrarily over the course of executions. In this paper, we conduct
\emph{parameterized analysis} of RBN.
We consider cubes, (infinite) sets of configurations
in the form of lower and upper bounds on the number of agents in each state, 
and we show that we can evaluate boolean combinations over cubes and reachability sets of cubes in \PSPACE.
In particular, reachability from a cube to another cube is a \PSPACE-complete problem. 

To prove the upper bound for this parameterized analysis, we prove some structural properties about the reachability sets and the  \emph{symbolic graph} abstraction of RBN, which
might be of independent interest. We justify this claim by providing two applications of these
results.
First, we show that the almost-sure coverability problem is \PSPACE-complete 
for RBN, thereby closing a complexity gap from a previous paper~\cite{Gandalf21}. 
Second,
we define a computation model using RBN, \`a la population protocols, called RBN protocols.
We characterize precisely the set of predicates that can be computed by such protocols.

%

\end{abstract}
%
%
\section{Introduction}

\input{intro}

\section{Preliminaries}

\input{preliminaries}

\section{Reachability sets of counting sets}\label{sec:main-sec}


\input{param-pbm-intro}

\input{symbolic-graph}
\input{all-properties}

\input{main-results}

\section{Application 1: Almost-sure coverability}\label{sec:almost-sure}

\input{almost-sure-cov}

\section{Application 2: Computation by RBN}\label{sec:computation}

\input{computation}


\section*{Acknowledgements}
We thank Nathalie Bertrand and Javier Esparza for many helpful discussions.

\bibliographystyle{splncs04}
\bibliography{bib}

\appendix
\input{appendix}

\end{document}

%% file: defs.tex
\usepackage{listings}
\lstdefinelanguage{pseudo}{morekeywords={init,with,or,if,then,else,fi,and,not,while,do,od,distinct,
    case, goto,local,algorithm, function, for, each, times, from, to,
    variables, procedure, recursive, return},
  morecomment=[l]{//}, morecomment=[s]{/*}{*/},
  mathescape=true,escapechar={@},
  basicstyle=\sffamily\small,
  commentstyle=\itshape\rmfamily\small,
  keywordstyle=\sffamily\bfseries\small
}
\usepackage{mathtools}
\usepackage{amsfonts,amssymb,amsmath, amsthm}

\usepackage{etoolbox}  
\usepackage{enumitem} 

\usepackage{stmaryrd}
\usepackage{bm}
\usepackage{complexity}

\usepackage{hyperref}
\usepackage{environ}
\usepackage{csvsimple}
\usepackage{longtable} 
\usepackage[normalem]{ulem} 

\usepackage{fourier-orns} 
\usepackage{wasysym} 
\usepackage{marvosym} 
\usepackage{bbding} 

\usepackage{etoolbox}  
\usepackage{enumitem} 

\usepackage{enumerate}
\usepackage{mathrsfs}
\usepackage{mathtools}
\usepackage{appendix}
\usepackage{comment}
\usepackage{tikz}
\usepackage{array}
\usetikzlibrary{calc}
\usepackage{multicol}

\usepackage{thm-restate} 

\usetikzlibrary{arrows,shapes,snakes,automata,backgrounds,petri,positioning}
\usetikzlibrary{fit,backgrounds} 
\definecolor{processblue}{cmyk}{0.96,0,0,0}

\newcommand{\nn}{\mathbb{N}}

\newcommand{\be}{\begin{enumerate}}
\newcommand{\ee}{\end{enumerate}}
\newcommand{\bc}{\begin{center}}
\newcommand{\ec}{\end{center}}

\newcommand{\bi}{\begin{itemize}}
\newcommand{\ei}{\end{itemize}}

\newcommand{\act}{\xrightarrow}
\newcommand{\gr}{\mathcal{G}}
\newcommand{\actsymb}[1]{\rightsquigarrow^{#1}}

\newcommand{\fin}{\mathit{fin}}

\newcommand{\prot}{\mathcal{P}}

\newcommand{\init}{\mathit{init}}

\newcommand\slice[2]{#1{\raise-.5ex\hbox{\ensuremath|}}_{#2}}



%% file: commands.tex
\newcommand{\defeq}{\stackrel{\scriptscriptstyle\text{def}}{=}}

\newcommand{\N}{\mathbb{N}}                    
\renewcommand{\vec}[1]{\bm{#1}}                
\newcommand{\set}[1]{\left\{#1\right\}}        
          
\newcommand{\Cv}{C_{\vec{v}}}
       
\newcommand{\multiset}[1]{\Lbag#1\Rbag}        
\newcommand{\supp}[1]{{\llbracket#1\rrbracket}}  
\renewcommand{\vec}[1]{\bm{#1}}                
\newcommand{\norm}[1]{\lVert#1\rVert}          

\renewcommand{\PP}{\mathcal{P}}                  
\newcommand{\cons}{\mathcal{C}}            
\newcommand{\stable}{\mathcal{ST}}        

\newcommand{\net}{\mathcal{N}}

\newcommand{\trans}[1]{\xrightarrow{#1}}       
\newcommand{\pre}{\mathit{pre}} 
\newcommand{\post}{\mathit{post}} 
\newcommand{\prestar}{\mathit{pre}^*}
\newcommand{\poststar}{\mathit{post}^*}

\renewcommand{\norm}[1]{\| {#1} \|}

\newcommand{\unorm}[1]{\|{#1}\|_u}
\newcommand{\lnorm}[1]{\|{#1}\|_l}
\newcommand{\sem}[1]{\llbracket{#1}\rrbracket} 

\newcommand{\cube}{\mathcal{C}}
\newcommand{\cC}{\Gamma}      
\newcommand{\cSet}{\mathcal{S}}  
\newcommand{\initcube}{\mathcal{I}} 

\newcommand{\RBN}{\mathcal{R}}

\newcommand{\ASMS}{\mathcal{P}}
\newcommand{\oper}{\mathtt{op}}
\newcommand{\data}{\textit{data}}

\makeatletter

\newcommand{\eqxrightarrow}[2]{%
  \mathop{%
    \vtop{%
      \m@th 
      \offinterlineskip 
      \ialign{%
        \hfil##\hfil\cr
        \rightarrowfill\cr
        \hphantom{$\scriptstyle\mskip8mu{#2}\mskip8mu$}\cr
        \vrule height0pt width 1.5em\cr
        $\scriptscriptstyle {#1}$\cr
      }%
    }%
  }\limits^{#2}%
}
\makeatother

%% file: intro.tex
Reconfigurable broadcast networks (RBN)~\cite{Delzanno12,AdHocNetworks} are a formalism for modelling distributed systems 
in which a set of anonymous, finite-state agents execute the same underlying protocol
and broadcast messages to their neighbors according to an underlying communication topology.
The communication topology is reconfigurable, meaning that the set 
of neighbors of an agent can change arbitrarily over the course of an execution. 
Parameterized verification of these networks concerns itself with proving
that a given property is correct, irrespective of the number of participating agents. 
Dually, it can  be viewed as the problem of finding an execution of some number
of agents which violates a given property. Ever since their introduction within this context~\cite{AdHocNetworks},
RBN have been studied extensively, with various results on (parameterized) reachability and coverability~\cite{Delzanno12,AdHocNetworks,Gandalf21,Liveness}, along with various extensions using probabilities 
and clocks~\cite{prob,probtime}. 

In this paper, we first consider the \emph{cube-reachability} problem for RBN, in which we are given
two (possibly infinite) sets of configurations $\cube$ and $\cube'$ (called \emph{cubes}), each of them defined 
by lower and upper bounds on the number of agents in each state, and we must decide
if there is a configuration in $\cube$ which can reach some configuration in $\cube'$.
The cube-reachability question covers parameterized reachability and coverability problems,
and as explained in~\cite{Gandalf21}, also covers the parameterized reachability problem
for a generalized model of RBN called \emph{RBN with leaders}. 
Moreover, a sub-problem of cube-reachability
 has already been studied for RBN in~\cite{Delzanno12}. The authors
show that this sub-problem is \PSPACE-complete. One of the  results in our paper
is that the entire cube-reachability problem is \PSPACE-complete, hence extending the sub-problem
considered in~\cite{Delzanno12}, while still retaining the same complexity upper bound.

In fact, our main result, which we call the \PSPACE{} Theorem, is a more general result. 
It subsumes the above result for cube-reachability and allows for more complex parameterized analysis of RBN. 
The \PSPACE{} Theorem roughly states
that any boolean combination of atoms
can be evaluated in \PSPACE{},
 where  an atom  is a finite union of cubes or the reachability set  of a finite union of cubes (i.e. $\poststar$ or $\prestar$). To prove the \PSPACE{} Theorem, we first consider the so called \emph{symbolic graph} of a RBN~(\cite{Delzanno12}, Section 5). We prove some structural properties about these graphs,  using
results from~\cite{Delzanno12}. Next, using these structural properties,
we show that the set of reachable configurations of a cube $\cube$ can be expressed as a finite union of cubes,  each  having a norm  exponentially bounded in the size of the given RBN and $\cube$.  
This result then allows us to give an on-the-fly exploration algorithm for proving the \PSPACE{} Theorem.

We believe that the \PSPACE{} Theorem and the results leading to it that we have proven in this paper have
further applications to problems concerning RBN. 
To justify this claim, we provide two applications.
First, we show that the almost-sure coverability problem for RBN is \PSPACE-complete, thereby
closing a complexity gap from a previous paper~(\cite{Gandalf21}, Section 5.3). Second, we define 
a computation model using RBN, called RBN protocols, which is similar in spirit to the population protocols model~\cite{First-Pop-Prot,Comp-Power-Pop-Prot}. 
We characterize precisely the set of predicates that can be computed using 
RBN protocols. 
This result generalizes the corresponding result for IO protocols, which 
are a sub-class of population protocols that can be simulated by RBN protocols,
as shown in~(\cite{Gandalf21}, Section 6.2).

Finally, by the reduction given in~(\cite{Gandalf21}, Section 4.2), our results on  cube-reachability and almost-sure coverability can  be transferred
to another model of distributed computation called asynchronous shared memory systems (ASMS),
giving a \PSPACE-completeness result for both of these problems. This solves an open problem
from~(\cite{ICALPPatricia}, Section 6).

To summarize, we have shown that many important parameterized problems of RBN can be solved in PSPACE, 
 that the sub-problem of the cube-reachability problem defined in \cite{Delzanno12} can be generalized while retaining the same upper bounds, 
 and that the almost-sure 
coverability problems for RBN and ASMS are \PSPACE-complete, thereby solving open problems
from \cite{Gandalf21,ICALPPatricia}. We believe that our other results
 might
be of independent interest, and we provide an application
by introducing RBN protocols and characterizing the set of predicates that they can compute. 

The paper is organized as follows. 
Section 2 contains preliminaries,  including the definition of RBN.
Section 3 defines the  symbolic graph of a RBN, and proves the properties of this graph needed to derive our main result.
Section 4 contains the main result that a host of parameterized problems over cubes,  including cube-reachability, is \PSPACE-complete for RBN.
Finally, Sections 5 and 6 give applications of our main results: 
Section 5 solves the complexity gap for the almost-sure coverability problem, and Section 6 introduces RBN protocols and characterizes their expressive power.
Due to lack of space, full proofs of some of the results can be found in the appendix.

%% file: preliminaries.tex

The definitions and notations in this section are taken from~\cite{Gandalf21}.

\subsection{Multisets}

A \emph{multiset} on a finite set \(E\) is a mapping \(C \colon E \rightarrow \N\), i.e. for any $e\in E$, \(C(e)\) denotes the number of occurrences of element \(e\) in \(C\).
We let $\mathbb{M}(E)$ denote the set of all multisets on $E$.
Let $\multiset{e_1,\ldots,e_n}$ denote the multiset $C$ such that $C(e)=|\{j\mid e_j=e\}|$.
We sometimes write multisets using set-like notation. 
For example, $\multiset{2 \cdot a,b}$ and $\multiset{a,a,b}$ denote the same multiset.
Given $e \in E$, we denote by $\vec{e}$ the multiset consisting of one occurrence of element
$e$, that is $\multiset{e}$. 
Operations on \(\N\) like addition or comparison are extended to multisets by defining them component wise on each element of \(E\).
Subtraction is allowed as long as each component stays non-negative.
We call $|C| \defeq\sum_{e\in E} C(e)$ the \emph{size} of $C$.

\subsection{Reconfigurable Broadcast Networks}

Reconfigurable broadcast networks (RBN) are networks consisting of finite-state, anonymous agents
and a communication topology which specifies for every pair of processes, whether or not there
is a communication link between them. During a single step, a single agent can broadcast a message
which is received by all of its neighbors, after which both the agent and its neighbors change their
state according to some transition relation. Further, in between two steps, the communication
topology can change in an arbitrary manner. For the problems that
we consider in this paper, it is  easier to forget
the communication topology and define the semantics of an RBN directly in terms of collections
of agents.

\begin{definition}
\label{def:rbn}
A \emph{reconfigurable broadcast network} is a tuple 
$\RBN = (Q, \Sigma,\delta)$ 
where $Q$ is a finite set of states,
$\Sigma$ is a finite alphabet 
and $\delta \subseteq Q \times \set{!a,?a \ | \ a \in \Sigma} \times Q$ is the transition relation.
\end{definition}

If $(p,!a,q)$ (resp. $(p,?a,q)$) is a transition in $\delta$, we will denote it by $p \act{!a} q$ (resp. $p \act{?a} q$).
A \emph{configuration} $C$ of a RBN $\RBN$ is a multiset over $Q$, which
intuitively counts the number of processes in each state. 
Given a letter $a\in \Sigma$ and two configurations $C$ and $C'$
we say  that there is a \emph{step} $C \trans{a} C'$
if there exists a multiset $\multiset{t, t_1, \ldots, t_k}$ of $\delta$ for some $k\ge 0$
satisfying the following: $t=p \trans{!a} q$, each $t_i =p_i \trans{?a} q_i$,
$C \ge \vec{p} + \sum_i \vec{p_i}$, and $C' = C - \vec{p} - \sum_i \vec{p_i} + \vec{q} + \sum_i \vec{q_i}$. 
We sometimes write this as $C \trans{t+t_1,\ldots, t_n} C'$ or $C \trans{a} C'$. 
Intuitively it means that a process at the state $p$ broadcasts the message $a$ and moves to $q$,
and for each $1 \le i \le k$, there is a process at the state $p_i$ which receives this message and moves to $q_i$.
We denote 
by $\trans{*}$ the reflexive and transitive closure of the step relation. 
A \emph{run} is then a sequence of steps.

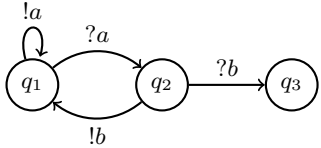
\begin{figure}
\input{fig-rbn}
\caption{An RBN $\RBN$ with three states.}
\label{fig:rbn}
\end{figure}

Let $\RBN= (Q, \Sigma,\delta)$ be an RBN.
Given configurations $C$ and $C'$, 
we say $C'$ is \emph{reachable} from $C$ if $C \trans{*} C'$.
We say $C'$ is \emph{coverable} from $C$ if there exists $C''$ such that $C \trans{*} C''$ and
$C'' \ge C'$.
The \emph{reachability} problem consists of deciding, 
given a RBN $\RBN$ and configurations $C,C'$, 
whether $C'$ is reachable from $C$ in $\RBN$.
The \emph{coverability} problem consists of deciding, 
given a RBN $\RBN$ and configurations $C,C'$, 
whether $C'$ is coverable from $C$ in $\RBN$.
Let $\cSet$ be a set of configurations. 
The \emph{predecessor set} of $\cSet$ is 
$\pre^*(\cSet) \defeq \{ C' | \exists C \in \cSet \, . \, C' \xrightarrow{*} C \}$, and the \emph{successor set} of $\cSet$ is
$\post^*(\cSet) \defeq \{ C | \exists C' \in \cSet \, . \, C' \xrightarrow{*} C \}$.

\begin{example}
\label{ex:rbn}
Figure \ref{fig:rbn} illustrates a RBN $\RBN=(Q, \Sigma,\delta)$ with $Q=\set{q_1,q_2,q_3}$. 
Configuration $\multiset{3 \cdot q_1}$ can reach $\multiset{2 \cdot q_1, q_3}$ in two steps. 
First, a process broadcasts $a$, the two other processes receive it and move to $q_2$. 
Then, one of the processes in $q_2$ broadcasts $b$ and moves to  $q_1$, while the other one receives $b$ and moves to $q_3$.
Notice that $\multiset{ q_3}$ is only coverable from a configuration $\multiset{k \cdot q_1}$ if $k\ge 3$.
\end{example}

\subsection{Cubes and Counting Sets}


Given a finite set $Q$, a \emph{cube} $\cube$ is a subset of $\mathbb{M}(Q)$ described 
by a lower bound $L \colon Q \rightarrow \N$ 
and an upper bound $U \colon Q \rightarrow \N \cup \{\infty\}$ 
such that $\cube = \{C : L \le C \le U\}$.
Abusing notation, we identify the set $\cube$ with the pair $(L,U)$.
Notice that since $U(q)$ can be $\infty$ for some state  $q$, a cube can contain an infinite number of configurations.
All the results in this paper are true irrespective of whether the constants in a given input cube
are encoded in unary or binary.

A
finite union of cubes $\bigcup_{i=1}^m (L_i,U_i)$ is called a \emph{counting constraint}
and the set of configurations $\bigcup_{i=1}^m \cube_i$ it describes is called a \emph{counting set}.
Notice that two different counting constraints may describe the same counting set.
For example, let $Q=\set{q}$ and let $(L,U)=(1,3)$, $(L',U')=(2,4)$, $(L'',U'')=(1,4)$. 
The counting constraints $(L,U)\cup(L',U')$ and $(L'',U'')$ define the same counting set.
It is easy to show (see also Proposition 2 of~\cite{EsparzaGMW18})
that counting constraints and counting sets are closed under Boolean operations.

\paragraph*{Norms.}
Let $\cube=(L,U)$ be a cube.
Let $\lnorm{\cube}$ be the the sum of the components of $L$.
Let $\unorm{\cube}$ be the sum of the finite components of $U$ if there are any, and $0$ otherwise.
The \emph{norm} of $\cube$ is the maximum of $\lnorm{\cube}$ and $\unorm{\cube}$, denoted by $\norm{\cube}$.
We define the norm of a counting constraint $\cC= \bigcup_{i=1}^m \cube_i$ as
$\norm{\cC} \defeq \displaystyle \max_{i\in [1,m]} \{ \norm{\cube_i} \}.$
The norm of a counting set $\cSet$ is the smallest norm of a counting constraint representing $\cSet$, that is, 
$\norm{\cSet} \defeq \displaystyle \min_{\cSet = \sem{\cC}} \{ \norm{\cC} \}$.
Proposition 5 of~\cite{EsparzaGMW18} entails the following results for the norms of the union, intersection and complement.
\begin{proposition}%
\label{prop:oponconf}
Let $\cSet_1, \cSet_2$ be counting sets.
The norms of the union, intersection and complement satisfy:
$ \norm{\cSet_1 \cup \cSet_2} \leq \max \{\norm{\cSet_1}, \norm{\cSet_2} \},
 \norm{\cSet_1 \cap \cSet_2} \leq \norm{\cSet_1} + \norm{\cSet_2}$, and  
 $\norm{\overline{\cSet_1}} \leq |Q| \cdot \norm{\cSet_1} + |Q|.$
\end{proposition}

\paragraph*{Reachability.}
The reachability problem can be generalized to the \emph{cube-reachability} problem which consists of deciding, given an RBN $\RBN$ and two cubes $\cube, \cube'$,
whether there exists configurations  $C \in \cube$ and $C' \in \cube'$ such that $C'$ is reachable from $C$ in $\RBN$.
If this is the case, we say $\cube'$ is reachable from $\cube$.
The \emph{counting set-reachability} problem asks, given an RBN $\RBN$ and two counting sets $\cSet, \cSet'$,
whether there exists cubes  $\cube \in \cSet$ and $\cube' \in \cSet'$ such that $\cube'$ is reachable from $\cube$ in $\RBN$.
We define \emph{cube-coverability}  and \emph{counting set-coverability} in an analoguous way.

\begin{remark}\label{remark:FSTTCS12}
	In the paper~\cite{Delzanno12}, the authors define a sub-class of the cube-reachability problem,
	which is called the \emph{unbounded initial cube-reachability} problem in~\cite{Gandalf21}. 
	More precisely, the sub-class considered in~\cite{Delzanno12} is the following:
	We are given an RBN and two cubes $\cube = (L,U)$ and $\cube' = (L',U')$ with
	the special property that $L(q) = 0$ and $U(q) \in \{0,\infty\}$ for every state $q$.
	We then have to decide if $\cube$ can reach $\cube'$. This problem was shown to be
	\PSPACE-complete~(\cite{Delzanno12}, Theorem 5.5), whenever the numbers in the input are given in unary. As we shall show later in this paper,
	the cube-reachability problem itself is in \PSPACE, even when the input numbers are encoded in binary, thereby generalizing the upper
	bound results given in that paper.
\end{remark}

%% file: fig-rbn.tex
\begin{center}
    \begin{tikzpicture}[->, thick]
      \node[place] (a1) {$q_1$};
	  \node[place] (a2) [right =of a1] {$q_2$};
	  \node[place] (a3) [right =of a2] {$q_3$};

      \path[->]
      (a1) edge[bend left = 40] node[above] {$?a$} (a2)
      (a2) edge[bend left = 40] node[below] {$!b$} (a1)
      (a2) edge node[above] {$?b$} (a3)
      (a1) edge[loop above] node[above] {$!a$} (a1)
      ;
    \end{tikzpicture}
\end{center}

%% file: param-pbm-intro.tex
In this section, we set the stage for proving the  main result of this paper. 
This main result is given in two stages: First, we show that 
given a RBN with state set $Q$ and a counting set $\cSet$, the set $\poststar(\cSet)$ is also a counting set and $\norm{\poststar(\cSet)} \le 2^{p(\norm{\cSet} \cdot |Q|)}$ where
$p$ is some fixed polynomial. 
Using this, we then prove that a host of cube-parameterized problems  for RBN can be solved in \PSPACE. 


The rest of this section is organized as follows: To prove the first result, we  recall the notion of a \emph{symbolic graph} of a RBN from~\cite{Delzanno12}.
In the symbolic graph, each node is a \emph{symbolic configuration} of the RBN, 
which intuitively represents
an infinite set of configurations in which 
the number of agents is fixed in some states, and arbitrarily big in the others.
Next, by exploiting the special structure of the symbolic graph, we prove some properties which allow us to show that 
whenever two nodes in this graph are reachable, they are reachable by a path having a  special structure.
Finally, using these properties and the connection between symbolic configurations and  configurations of the RBN,
we prove the desired first result. Once we have shown the first result, we then show how the \PSPACE{} Theorem can be obtained from it.

Throughout this section, we fix an RBN $\RBN = (Q,\Sigma,\delta)$.

%% file: symbolic-graph.tex
\subsection{Symbolic graph}

In this subsection, we recall the notion of a \emph{symbolic graph} of an RBN from~\cite{Delzanno12}. 
Here, for the sake of convenience, we define it in a slightly different way, but the underlying notion is the same as~\cite{Delzanno12}. Throughout this subsection and the next,
we fix a number $k \in \nn$.

The \emph{symbolic graph of index} $k$ associated with the RBN $\RBN$ is an edge-labelled graph
$\gr_k = (N,E,L)$ where $N = \mathbb{M}_k(Q) \times 2^Q$ is the set of nodes. 
Here $\mathbb{M}_k(Q)$ denotes
the set of multisets on $Q$ of size at most $k$. 
$E$ is the set of edges 
and $L : E \to \Sigma$ is the labelling function.
Each node of $\gr_k$ is also called a \emph{symbolic configuration}.
Intuitively, in each symbolic configuration $(v,S)$, the multiset $v$ (called the \emph{concrete part}) is used to 
keep track of a fixed set of at most $k$ agents, and the subset $S$ (called the \emph{abstract part}) is used to keep track of the \emph{support}
of the remaining agents. 

Let $\theta = (v,S)$ and $\theta' = (v',S')$ be two symbolic configurations. There is an edge labelled by $a$ between $\theta$ and $\theta'$ if and only if the following is satisfied: There exists a transition $(q,!a,q') \in \delta$ such that at least one 
of the following two conditions holds
\begin{itemize}
	\item \textit{(Broadcast from $v$)} There exists a multiset of transitions $\multiset{(p_1,?a,p_1'),\dots,\allowbreak (p_l,?a,p_l')}$ such that 
	$v' = v - \sum_i \vec{p_i} + \sum_i \vec{p_i'} - \vec{q} + \vec{q'}$, and for each $q_s \in Q$:
	\begin{itemize}
		\item If $q_s \in S' \setminus S$ then there exists $q'_s \in S$ and $(q'_s, ?a, q_s) \in R$,
		\item If $q_s \in S \setminus S'$ then there exists $q'_s \in S'$ and $(q_s, ?a, q'_s) \in R$.
	\end{itemize}
	\item \textit{(Broadcast from $S$)} There exists a multiset of transitions $\multiset{(p_1,?a,p_1'),\dots,\allowbreak (p_l,?a,p_l')}$ such that
	$v' = v - \sum_i \vec{p_i} + \sum_i \vec{p_i'}$, $q \in S, q' \in S'$, and for each $q_s \in Q \setminus \{q,q'\}$:
	\begin{itemize}
		\item if $q_s \in S' \setminus S$ then there exists $q'_s \in S$ and $(q'_s, ?a, q_s) \in R$,
		\item if $q_s \in S \setminus S'$ then there exists $q'_s \in S'$ and $(q_s, ?a, q'_s) \in R$.
	\end{itemize}
\end{itemize}

An edge labelled by $a$ between $\theta$ and $\theta'$ is denoted by 
$\theta \rightsquigarrow^a_{\gr_k} \theta'$.
The relation $\rightsquigarrow^*_{\gr_k}$ is the reflexive and transitive closure of $\rightsquigarrow_{\gr_k} := \cup_{a \in \Sigma}\rightsquigarrow^a_{\gr_k}$. Whenever the index $k$ is clear, we will drop
the subscript $\gr_k$ from these notations.

\begin{remark}
\label{rmk:symbgraph-reach}
Let $\theta = (v,S),\theta' = (v',S')$ be two symbolic configurations.
By construction, $\theta$ can only reach $\theta'$ if $|v|=|v'|$.
\end{remark}

To give an intuition behind the edges in $\gr_k$, recall the intuition that in a symbolic configuration, the concrete part is used to keep track
of a fixed set of at most $k$ processes and the abstract part is used to keep track of the support of the remaining processes. The first 
condition for the existence of an edge asserts the following: 1) In the concrete part, some process broadcasts the message $a $ and some subset of processes receive $a$, 2) In the abstract part, any new state added or any old state deleted comes because of receiving $a$.
The second condition asserts exactly the same, except we now require the process broadcasting the message $a$ to be from the abstract part.

The symbolic graph of index $k$ can be thought of as an abstraction of the set of configurations of $\RBN$, where only a fixed number of processes
are explicitly represented and the rest are abstracted by means of their support alone. To formalize this, given a symbolic configuration
$\theta = (v,S)$, we let $\supp{\theta}$ denote the following (infinite) set of configurations: $C \in \supp{\theta}$ if and only if
$C(q) = v(q)$ for $q \notin S$ and $C(q) \ge v(q)$ for $q \in S$.

\begin{figure}
\input{fig-G0}
\caption{Symbolic graph $\gr_0$ of index $0$ of the RBN of Example \ref{ex:rbn}.}
\label{fig:G0}
\end{figure}
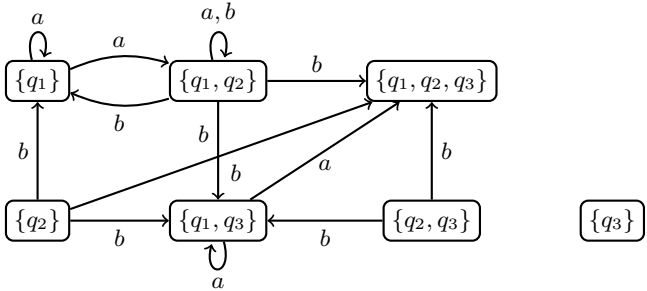

\begin{example}
The symbolic graph $\gr_0$ of index $0$ of the RBN of Example \ref{ex:rbn} is illustrated in Figure \ref{fig:G0}.
At this index, the graph only keeps track of a subset $S \subseteq Q$, and the edges correspond to broadcasts from $S$.
Consider the edges from $\set{q_1}$. 
The self-loop corresponds to a broadcast of $a$ that is not received.
The edge to $\set{q_1,q_2}$ corresponds to a broadcast of $a$ received by at least one process in $q_1$.
There is no edge from $\set{q_3}$ because there is no broadcast transition from $q_3$.
\end{example}

We then have the following lemma, which asserts that runs between two configurations in an RBN
 induce corresponding runs in the symbolic graph. The proof of the lemma is easily obtained
from the definition of the symbolic graph.

\begin{restatable}{lemma}{LemmaRuns}\label{lm:runs}
	Let $C,C'$ be two configurations of $\RBN$ such that $C \act{a} C'$.
	Then, for every $\theta$ such that $C \in \supp{\theta}$, 
	there exists $\theta'$ such that $C' \in \supp{\theta'}$
	and $\theta \actsymb{a} \theta'$.
\end{restatable}

%% file: fig-G0.tex
\begin{center}
    \begin{tikzpicture}[->, thick, node distance=1.3cm]
      \node[rounded corners=1mm,draw]  (a1) {$\set{q_1}$};
      \node[rounded corners=1mm,draw] (a2) [right =of a1] {$\set{q_1,q_2}$};
      \node[rounded corners=1mm,draw] (a3) [right =of a2] {$\set{q_1,q_2,q_3}$};
      \node[rounded corners=1mm,draw] (a4) [below =of a1] {$\set{q_2}$};
      \node[rounded corners=1mm,draw] (a5) [right =of a4] {$\set{q_1,q_3}$};
      \node[rounded corners=1mm,draw] (a6) [below =of a3] {$\set{q_2,q_3}$};
      \node[rounded corners=1mm,draw] (a7) [right =of a6] {$\set{q_3}$};

      \path[->]
      (a1) edge[loop above] node[above] {$a$} (a1)
      (a1) edge[bend left = 20] node[above] {$a$} (a2)
      (a2) edge[bend left = 20] node[below] {$b$} (a1)
      (a4) edge node[left] {$b$} (a1)
      (a2) edge node[above] {$b$} (a3)
      (a2) edge node[above left] {$b$} (a5)
      (a2) edge[loop above] node[above] {$a,b$} (a2)
      (a4) edge node[below] {$b$} (a5)
      (a4) edge node[below right] {$b$} (a3)
      (a5) edge node[below] {$a$} (a3)
      (a5) edge[loop below] node[below] {$a$} (a5)
      (a6) edge node[right] {$b$} (a3)
      (a6) edge node[below] {$b$} (a5)
      ;
    \end{tikzpicture}
\end{center}

%% file: all-properties.tex
\subsection{Properties of the symbolic graph}

In this subsection, we prove some properties of the symbolic graph (of any index $k$). 
The first two properties that we prove exhibit some structural properties on the paths of the symbolic graph.
The next two properties relate paths over the symbolic graph to runs over the configurations of the given RBN.
These four properties will ultimately lead us to prove our two main contributions in the next section.

\subsubsection*{First property: Monotonicity.}

Let $k \in \nn$ and let $\gr_k$ be the symbolic graph of index $k$ associated with $\RBN$. 
The first key property of $\gr_k$ is the following 
property, which we call \emph{monotonicity}. 

\begin{proposition}\label{prop-monotone}
	Let $\theta = (v,S)$ and $\theta' = (v', S')$ be symbolic configurations of $\gr_k$. Then the following are true:
	\begin{itemize}
		\item If $Z \subseteq S$ and $\theta \actsymb{a} \theta'$, then $(v, S) \actsymb{a} (v', Z \cup S')$. 
		\item If $Z \subseteq Q$ and $\theta \actsymb{a} \theta'$, then $(v, Z \cup S) \actsymb{a} (v', Z \cup S')$.
	\end{itemize}	
\end{proposition}

\begin{proof}
	The two points follow immediately from the definition of $\actsymb{a}$.
\end{proof} 

\subsubsection*{Second property: Normal Form.}

To state the second property, we first need a small definition.
\begin{definition}
	Let $(v_0, S_0) \actsymb{} \cdots \actsymb{} (v_m, S_m)$ a path in $\gr_k$.
	A pair of indices $0 \le i < j \le m$ is called a \emph{bad pair} if $(S_i \setminus S_{i+1}) \cap S_j \neq \emptyset$. 
	A path is said to be in \emph{normal form} if it contains no bad pairs, i.e., 
	for all $0 \le i < m$ and any $j > i$,
	$(S_i \setminus S_{i+1}) \cap S_j = \emptyset$. 
\end{definition}
Intuitively, a path is in normal form if during each step, the states that disappear from the
abstract part never reappear again.  The following lemma asserts that whenever there is a path between two symbolic configurations,
then there is a path between them that is in normal form. 

\begin{restatable}{lemma}{LemmaNF}\label{lemma-nf}
	Let $\theta, \theta'$ be symbolic configurations of $\gr_k$ such that there is a path between $\theta$ and $\theta'$ of length $m$.
	Then, there is a path in normal form between $\theta$ and $\theta'$ of length $m$.
\end{restatable}	

\begin{proof}[Proof Sketch]
	Let $\theta = \theta_0 \actsymb{} \theta_1 \actsymb{} \theta_2 \actsymb{} \dots \theta_{m-1} \actsymb{} \theta_m = \theta'$ be the path between $\theta$ and $\theta'$.
	We proceed by induction on $m$. The claim is clearly true for $m = 0$. Suppose $m > 0$ and the claim is true for $m-1$. By induction hypothesis, we can assume 
	that the path $\theta_0 \actsymb{} \theta_1 \actsymb{} \dots \actsymb{} \theta_{m-1}$ is already in normal form.
	
	Let each $\theta_i = (v_i,S_i)$. Let $l$ be the number of bad pairs in the path between $\theta_0$ and $\theta_m$. If $l = 0$, then the path is already in normal form and we are done.
	Suppose $l > 0$ and let $(w,w')$ be a bad pair. Since the path between $\theta_0$ and $\theta_{m-1}$ is already in normal form, it has to be the case that $w' = m$.
	Hence, we have $Z := (S_{w} \setminus S_{w+1}) \cap S_m \neq \emptyset$.
	
	By Proposition~\ref{prop-monotone}, the following is a valid path: $(v_{w},S_{w}) \actsymb{} (v_{w+1},S_{w+1} \cup Z) \actsymb{} (v_{w+2},S_{w+2} \cup Z) \dots (v_{m-1}, S_{m-1} \cup Z)\actsymb{} (v_m, S_m \cup Z) = (v_m,S_m)$. 
	Let $\theta'_j := \theta_j$ if $j \le w$ and $(v_j,S_j \cup Z)$ otherwise. 
	Hence, we get a path $\theta'_0 \actsymb{} \theta'_1 \actsymb{} \dots \theta'_{m-1} \actsymb{} \theta'_m$.
	
	Let each $\theta'_e = (v_e',S_e')$ and let $0 \le i < j \le m-1$. By a case analysis on where $i$ and $j$ are relative to the index $w$, we can prove that $(S_i' \setminus S_{i+1}') \cap S_j' = \emptyset$.
	Having proved this, it is then clear by construction, 
	that this new path from $\theta'_0 := \theta_0$ to $\theta'_m := \theta_m$ has at most $l-1$ bad pairs only.
	Hence, we now have a path from $\theta_0$ to $\theta_m$ such that the prefix of length $m-1$ is in normal form and the number of bad pairs has been strictly reduced to $l-1$.
	Repeatedly applying this procedure leads to a path in normal form between $\theta_0$ and $\theta_m$.
\end{proof}

\subsubsection*{Third property: Refinement.}

Before we state the third property, we need a small definition. Recall that, given a symbolic configuration $\theta = (v,S)$, 
the set $\supp{\theta}$ denotes the set of configurations $C$ such that $C(q) = v(q)$ if $q \notin S$ and $C(q) \ge v(q)$ otherwise.
The following definition refines the set $\supp{\theta}$.
\begin{definition}
	Given a symbolic configuration $\theta = (v,S)$ and a number $N \in \nn$, let $\supp{\theta}_N$ denote the set of configurations
	$C$ such that $C(q) = v(q)$ if $q \notin S$ and $C(q) \ge v(q) + N$ otherwise. Note that $\supp{\theta} = \supp{\theta}_0$.
\end{definition}

This definition along with the above two properties now enable us to prove the  third property. It roughly states that if a symbolic configuration $\theta'$ can be reached from another symbolic configuration $\theta$, then 
there is a ``small'' $N$ such that any configuration in $\supp{\theta'}_N$ can be reached from some configuration in $\supp{\theta}$. 

\begin{restatable}{theorem}{TheoremUpperBoundN}\label{theorem-upperboundN}
	Let $\theta, \theta'$ be symbolic configurations of $\gr_k$ such that $\theta \actsymb{*} \theta'$.
	Then there exists $N \leq k \times (2k)^{|Q|} \times (|Q|+1)^{|Q|+1} + 1$ such that for all $C' \in \supp{\theta'}_N$, there exists $C \in \supp{\theta}$ such that $C \act{*} C'$.	
\end{restatable}
 
\begin{proof}[Proof Sketch]
	Suppose $\theta \actsymb{*} \theta'$. If the length of the path is 0, then there is nothing to prove.
	Hence, we restrict ourselves to the case when the length of the path is bigger than 0.
	By Lemma~\ref{lemma-nf}, there is a path in normal from from $\theta$ to $\theta'$ (say)
	$\theta = \theta_0 \actsymb{} \theta_1 \actsymb{} \theta_2 \dots \theta_{m-1} \actsymb{} \theta_m = \theta'$ with each
	$\theta_i := (v_i,S_i)$. 
	
	Let $N_0 = 0$ and let $N_i = (N_{i-1}+1) \cdot (|S_{i-1} \setminus S_i|+1)$ for every $1 \le i \le m$. 
	In Lemma 5.3 of~\cite{Delzanno12} (more precisely in its proof, in Lemma 6 of the long version \cite{Delzanno12Long}), the following fact has been proved: 
	\begin{quote}
		For every $1 \le i \le m$ and for every $C' \in \supp{\theta_i}_{N_i+1}$, there exists $C \in \supp{\theta_{i-1}}_{N_{i-1}+1}$ such that
		$C \act{*} C'$.
	\end{quote}
	
	This immediately proves that for all $C' \in \supp{\theta'}_{N_m + 1}$, there exists $C \in \supp{\theta}$ such that $C \act{*} C'$. If we prove $N_m \le k \times (2k)^{|Q|} \times (|Q|+1)^{|Q|+1}$, then the proof of the theorem will be complete.

	Notice that if $(v,\emptyset) \actsymb{} (v',S')$ is an edge in $\gr_k$ then $S' = \emptyset$.
	This fact, along with the definition of a path in normal form, allows us to easily conclude 
	that the number of indices $i$ such that $|S_{i-1} \setminus S_i| > 0$
	is at most $|Q|$. 
	It then follows that except for at most $|Q|$ indices, each index $N_i$ is obtained from $N_{i-1}$ by simply adding 1
	and in the remaining indices, $N_i$ is obtained from $N_{i-1}$ by adding 1 and then multiplying by a number which is at most $|Q|+1$. 
	Using this, we can deduce that the maximum value for $N_m$ is at most 
	$(m-|Q|+1)|Q|(|Q|+1)^{|Q|}$. 
	Since $m$ is itself the length of the path between $\theta_0$ and $\theta_m$,
	$m$ is upper bounded by the number of symbolic configurations in $\gr_k$ which is at most $k \times k^{|Q|} \times 2^{|Q|}$.
	Overall we get that $N_m \le k \times (2k)^{|Q|} \times (|Q|+1)^{|Q|+1}$.
\end{proof}

\begin{remark}
	A similar result was proved in Lemma 5.3 of~\cite{Delzanno12}, but there it was just stated that
	there exists an $N$ satisfying this property. Moreover from the proof of that lemma, 
	only a doubly exponential bound on $N$ could be inferred.
\end{remark}

\subsubsection*{Fourth property: Compatibility.}

To describe the fourth property, we need the following notion of order on configurations, \emph{relative} to a given symbolic configuration.

\begin{definition}
	\label{def:symbconfig-order}
	Let $\theta = (v,S)$ be a symbolic configuration, and let $C, C'$ be two configurations of $\RBN$.
	We define an order $\preceq_{\theta}$ such that $C \preceq_{\theta} C'$ if and only if $C ,C' \in \doublec{ \theta}$, and $\forall q \in S$, $C(q) \leq C'(q)$.
\end{definition}

This definition enables us to state our next property, which we dub \emph{compatibility}. It intuitively says that
the order that we have defined is, in some sense, compatible with the edges of the symbolic configurations.

\begin{lemma}
	\label{lemma-poststar}
	Let $\theta $ be a symbolic configuration of $\gr_k$, 
	and let $C, C'$ be two configurations of $\RBN$.
	If $C \in \doublec{\theta}$ and $C \act{*} C'$,
	then there exists a symbolic configuration $\theta'$ such that 
	1) $C' \in \doublec{\theta'}$, 2) $\theta \actsymb{*} \theta'$ and 3) 
	for all $C'_1$ such  that $C'_1 \succeq_{\theta'} C'$,  
	there exists $C_1 \in \doublec{\theta}$ such that $C_1 \act{*} C'_1$.
\end{lemma}

\begin{proof}
	Let $\theta$ be a symbolic configuration and $C, C'$ be configurations such that $C\in \doublec{\theta}$ and $C\act{*}C'$.
	Let $C = C_0 \act{} \cdots\act{} C_{m-1} \act{} C_m = C'$ denote the run between $C$ and $C'$.
	We prove the property by induction on $m$. 
	For $m =0$, we have $C= C'$. The property is easily seen to hold with $\theta' = \theta$.
	
	Suppose now that $m \geq 1$, and that the property holds for all $n \leq m$. 
	By induction hypothesis, for the configuration $C_{m-1}$, there exists a symbolic configuration $\theta_{m-1}$ satisfying the  property, in particular $\theta \actsymb{*} \theta_{m-1}$.
	Since $C_{m-1} \act{a} C_m$ for some $a \in \Sigma$, by Lemma \ref{lm:runs},  
	there exists a symbolic configuration $\theta_m$ such that $C_m \in \doublec{\theta_m}$, and $\theta_{m-1} \actsymb{a} \theta_m$.
	Using $\theta \actsymb{*} \theta_{m-1}$, we obtain that $\theta \actsymb{*} \theta_m$.
	
	Let $\theta_{m-1} =(v_{m-1},S_{m-1})$ and $\theta_m = (v_m, S_m)$.
	Let $C'_m \in \doublec{\theta_m}$ be such that $C'_m \succeq_{\theta_m} C_m$. 
	We will construct a configuration $C'_{m-1}\in \doublec{\theta_{m-1}}$ such that $C'_{m-1} \succeq_{\theta_{m-1}} C_{m-1}$ and $C'_{m-1} \act{*} C'_m$. If we construct such a configuration, then by induction
	hypothesis, there is a $C_1 \in \doublec{\theta}$ such that $C_1 \act{*} C'_{m-1} \act{*} C'_m$, which will conclude the proof.

	Let $C'_{m-1}(q) =   C_{m-1}(q)$	for all $q \not \in S_{m-1}$.
	To define $C'_{m-1}$ on $S_{m-1}$, 
	we first define a  mapping $\emph{pred}$ from states in $S_m$ to states of $ S_{m-1} \cup  \overline{S_{m-1}}=Q$ as follows. Given $q'\in S_m$:
	\begin{itemize}
		\item If $q' \in S_{m-1}$, $\emph{pred}(q') = q'$;
		\item If $q' \not \in S_{m-1}$, by definition of edges in the symbolic graph, there exists $q \in S_{m-1}$ such that $(q,?a,q')$ is a transition. 
		Then $\emph{pred}(q') = q$ for one (arbitrary  but fixed) such $q$.
	\end{itemize} 
	By definition, $ C'_m(q)= C_m(q) $ for all $q \not \in S_m$. 
	For all $q \in S_m$, let $n_q =   C'_m(q) -   C_m(q)$. 
	Intuitively, we want to place these $n_q$ processes in the right places of $C'_{m-1}$ so that $C'_{m-1} \act{} C'_m$. 	
	For all $q \in S_{m-1}$, let $  C'_{m-1}(q) =   C_{m-1}(q) + \sum_{q' \in S_m,\emph{pred}(q') = q} n_{q'}$. 
	By definition, $C'_{m-1} \succeq_{\theta_{m-1}} C_{m-1}$.
	So all that remains is to prove  that $C'_{m-1} \act{*} C'_m$.
	
	Let $C_{m-1} \trans{t+t_1,\dots,t_n} C_m$ where $t = (p,!a,p')$ and
	each $t_i = (p_i,?a,p_i')$. If we let $S_m \setminus S_{m-1} = \{q_1',\dots,q_w'\}$, then
	by definition there is a transition $t_i' := (pred(q_i'),?a,q_i')$ for each $i$.
	Additionally,  $C'_{m-1}(pred(q_i')) \ge C_{m-1}(pred(q_i')) + n_{q_i'} $. 
	This allows us to
	do $C'_{m-1} \trans{t+t_1,\dots,t_n,n_{q_1'} \cdot t_1', n_{q_2'} \cdot t_2', \dots, 
	n_{q_w'} \cdot t_w'} C'_{m}$, which concludes the proof.

\end{proof}

%% file: main-results.tex
\newcommand{\symbcube	}{\Delta_\cube}
\newcommand{\cubemin}{\cube_C^\theta}
\section{The \PSPACE{} Theorem}

In this section, we prove our two main contributions. 
First, we show that given a cube $\cube$, $\poststar(\cube)$ is a counting set of bounded size. 
Using this, we show our main result:
any boolean combination of atoms
can be evaluated in \PSPACE{},
 where  an atom  is a counting set or the reachability set of a counting set.
We call this the \emph{\PSPACE \ Theorem}. 
The intuition behind the \PSPACE{} Theorem is that 
the norms of the counting sets obtained
by such combinations are ``small'', and so we only need to examine small configurations to verify them, thus yielding a \PSPACE \ algorithm for checking correctness.
In particular, the \PSPACE{} Theorem will show that the cube-reachability problem is in \PSPACE.
We fix an arbitrary RBN $\RBN=(Q, \Sigma, \delta)$ for the rest of the section.

We start by drawing links between cubes and symbolic configurations.
\begin{itemize}
\item Given a symbolic configuration $\theta=(v,S)$,
we let $\cube_\theta$ be the cube $(L,U)$ 
where $L=v$, and $U(q)=v(q)$ if $q \notin S$ and $U(q)=\infty$ otherwise.
Then $\cube_\theta = \doublec{\theta}$.
\item Given a cube $\cube=(L,U)$,
we define $\symbcube$ to be the set of symbolic configurations
$\theta=(v,S)$ with $S = \set{q \ | \ U(q)= \infty}$ and $L(q) \le v(q) \le U(q)$ if $q \notin S$ and $v(q)=L(q)$ otherwise. 
Then $\doublec{\symbcube} = \cube$.
\end{itemize}
%


Notice that the set $\symbcube$ is included in the symbolic graph of index $2 \norm{\cube}$.
Indeed, if $\cube = (L,U)$ and $(v,S) \in \symbcube$, then $|v| \le |L| + |U_f|$ where
$U_f(q)=0$ if $U(q)=\infty$ and $U_f(q)=U(q)$ otherwise. Since 
$\norm{\cube} = \max(|L|,|U_f|)$, we have the desired result.
By Remark \ref{rmk:symbgraph-reach}, 
we know that symbolic configurations in the graph of index $2\norm{\cube}$ 
can only reach symbolic configurations which are also in  the graph of index $2\norm{\cube}$.

\begin{lemma}\label{lm:index}
Given a cube $\cube$, the sets $\symbcube$ and $\poststar(\symbcube)$
are included in the symbolic graph of index $2 \norm{\cube}$.
\end{lemma}

There are only a finite number of symbolic configurations in the graph of a given index.
Therefore $\poststar(\symbcube)$ is a finite set of symbolic configurations $\theta$.
It follows that $\doublec{\poststar(\symbcube)}$ is the finite union of the cubes $\cube_\theta$, and thus a counting set.

Unfortunately, it is in general not the case that $\poststar(\cube)= \doublec{\poststar(\symbcube)}$, which would close our argument.
However, we will show that for each symbolic configuration $\theta$ in $\poststar(\symbcube)$,
there is a counting set $\cSet_\theta \subseteq \doublec{\theta}$ 
such that the finite union of these counting sets is equal to $\poststar(\cube)$.
This will then show our first important result, namely that the reachability set of a counting set
is also a counting set with ``small'' norm.

\begin{theorem}
\label{thm:poststar}
Let $\cube$ be a cube. \chana{how detailed do we want this bound?}
Then $\poststar(\cube)$ is a counting set and
$$
\norm{\poststar(\cube)} \in O((\norm{\cube} \cdot |Q|)^{|Q|+2})
$$
The same holds for $\prestar$ by using the given RBN with reversed transitions.
\end{theorem}
\begin{proof}
We start by defining a counting set $\mathcal{M}$ of configurations, 
which we will then prove to be equal to $\poststar(\cube)$.\bala{Check my changes}\chana{done, some typos corrected}
Given a symbolic configuration $\theta$ of $\poststar(\symbcube)$,
we define the set $\min(\theta,\cube)$ to be the set of 
configurations $C \in \doublec{\theta}$ such that $C$ is minimal 
for the order $\preceq_{\theta}$ over the configurations of $\poststar(\cube)$, i.e.
$$
\min(\theta,\cube) = \min_{\preceq_{\theta}} \set{C \in \doublec{\theta} \ | \ C \in \poststar(\cube)}
$$
We can now define $\mathcal{M}$ to be the following set
$$
\mathcal{M}= \bigcup_{\theta \in \poststar(\symbcube)} \ \bigcup_{C \in \min(\theta,\cube)} \cubemin,
$$
where $\cubemin$ is the cube $\cube_{(C,S)}$ for $S$ such that $\theta=(v,S)$.
Since $\mathcal{M}$ is  a finite union of cubes, it is a counting set.

We show that $\poststar(\cube) \subseteq \mathcal{M}$.
Let $C \in \poststar(\cube)$.
There exists $C_0 \in \cube$ such that $C_0 \act{*} C$, 
and there exists $\theta_0 \in \symbcube$ such that $C_0 \in \doublec{\theta_0}$.
Applying Lemma \ref{lm:runs},
we obtain the existence of $\theta \in \poststar(\theta_0)\subseteq \poststar(\symbcube)$ such that $C \in \doublec{\theta}$.
Now, there exists a configuration $C' \in \min(\theta,\cube)$ 
such that $C' \preceq_{\theta}C$.
By definition of $\cube_{C'}^\theta$, $C$ is in $\cube_{C'}^\theta$ and thus in $\mathcal{M}$.

Now we show that $\mathcal{M} \subseteq \poststar(\cube)$. 
Let $C \in \mathcal{M}$. By definition, there must be a 
symbolic configuration $\theta \in \poststar(\symbcube)$ and a configuration
$C' \in \poststar(\cube)$ such that $C' \preceq_{\theta}C$.
By the Compatibility Lemma (Lemma \ref{lemma-poststar}), $C$ is in $\poststar(\cube)$ as well.

All that remains is to bound the norm of $\mathcal{M}$. To do this, let 
$\theta = (v,S) \in \poststar(\symbcube)$ and let $C \in \min(\theta,\cube)$. 
If we bound the norm of $\cubemin$ by the desired quantity, then the proof will be complete.
Noticing that $\norm{\cubemin} = |C|$, it suffices to bound $|C|$ by
the desired quantity, which is what we shall do now.

By Theorem \ref{theorem-upperboundN} and Lemma \ref{lm:index},
there exists an $N\le 2\norm{\cube} \times (4\norm{\cube})^{|Q|} \times (|Q|+1)^{|Q|+1}$ such that
$\doublec{\poststar(\symbcube)}_N \subseteq \poststar(\doublec{\symbcube}) = \poststar(\cube)$.
By definition of $C$, there must be a smallest $N'$ such that $C(q) \le v(q) + N'$ for every state  $q$. 
If $N' > N$, then let $C_N$ be the configuration given by
$C_N(q) = \min(C(q),v(q)+N)$. 
We get that $C_N \in \doublec{\theta}_N \subseteq
\doublec{\poststar(\symbcube)}_N \subseteq \poststar(\cube)$, and so $C_N \preceq_{\theta}
C$ and $C_N \in \poststar(\cube)$, which is a contradiction to the
minimality of $C$. Hence $N' \le N$ and so $|C| \le |v| + |Q| \cdot N$.
Since $\theta = (v,S)$ is in $\poststar(\symbcube)$, by Lemma~\ref{lm:index}, we have
that $|v| \le 2 \norm{\cube}$. Substituting the upper bounds for $|v|$ and $N$ in the
inequality $|C| \le |v| + |Q| \cdot N$ then gives the required upper bound for $|C|$,
thereby finishing the proof.


This result also holds for $\prestar(\cube)$. 
If $\RBN = (Q,\Sigma,R)$ is the given RBN, consider
the ``reverse" RBN $\RBN_r$, defined as $\RBN=(Q,\Sigma,R_r)$ where $R_r$ has a transition $(q, \star a, q')$ for $\star \in \set{!,?}$ 
if{}f $R_r$ has a transition $(q', \star a, q)$. 
Notice that $\RBN_r$ is still an RBN and that $\post^*(\cube)$ in $\RBN$ is equal to $\pre^*(\cube)$ in $\RBN_r$.
\end{proof}

Recall that counting sets are closed under boolean operations. With the above theorem, plus the fact that counting sets are finite unions of cubes, we obtain the following closure result.

\begin{corollary}[Closure]
\label{coro:closure}
Counting sets are closed under $\poststar$, $\prestar$ and boolean operations.
\end{corollary}

We are now ready to show our  main result, the \PSPACE \ Theorem.
We show that there exist \PSPACE{} algorithms to evaluate boolean combinations over counting sets and  reachability set of counting sets.
This result and its proof are adapted from a similar result for population protocols in~\cite{DC}. 

Given a counting constraint $\Gamma$, we let $[\Gamma]$ denote the counting set described by $\Gamma$. To state our result, we first define some ``nice'' expressions.
\begin{definition}
A \emph{nice} expression is any expression that is constructed by the following syntax:
$$E := \Gamma \ | \ \poststar(\Gamma) \ | \ \prestar(\Gamma) \ | \ E \cap E \ | \ E \cup E \ | \ \overline{E}$$
where $\Gamma$ is any counting constraint. 

If $E$ is a nice expression, then the \emph{size} of $E$, denoted by $|E|$, is defined as follows:
\begin{itemize}
\item If $E = \Gamma$ or $\poststar(\Gamma)$ or $\prestar(\Gamma)$, then $|E| = 1$;
\item If $E = E_1 \cup E_2$ or $E = E_1 \cap E_2$, then $|E| = |E_1| + |E_2|$;
\item If $E = \overline{E_1}$, then $|E| = |E_1| + 1$.
\end{itemize}
 The set of configurations that is described by a nice expression $E$ can be defined in a straightforward manner, and is denoted as $[E]$.
\end{definition}

Notice that any nice expression $E$ is a counting constraint, and $[E]$ is a  counting set, by the Closure Corollary \ref{coro:closure}.

\begin{theorem}[\PSPACE{} Theorem]
\label{thm:pspace-lemma}
%
Let $E$ be a nice expression and let $N$ be the maximum norm of the counting constraints appearing in $E$. 
Then $[E]$ is a counting set of norm  at most exponential in $N, |E|$ and $|Q|$.
Further, the membership and emptiness problems for $[E]$ are in \PSPACE{}.

\end{theorem}
\begin{proof}
		Recall that $[E]$ is a counting set , by the Closure Corollary (Corollary~\ref{coro:closure}).
        The exponential bounds for the norms follow immediately from 
        Proposition~\ref{prop:oponconf} and Theorem~\ref{thm:poststar}.
        The membership complexity for union, intersection and 
        complement is easy to see. 
        Without loss of generality it suffices to prove that membership in
        $\poststar(\Gamma)$ is in \PSPACE, where $\Gamma$ is a counting constraint.

        By Savitch's Theorem \NPSPACE=\PSPACE, so we provide a nondeterministic algorithm.
        Given $(C,\Gamma)$, we want to decide whether $C\in \poststar(\Gamma)$.
        The algorithm
        first guesses a configuration $C_0 \in \Gamma$ of the same size as $C$,
        verifies that $C_0$ belongs to $\Gamma$, 
        and then simply guesses an execution starting at $C_0$, step by step.
        The algorithm stops if either the configuration reached at some step is $C$, 
        or if it has guessed more steps than the number of configurations of size $|C|$.
        This concludes the discussion regarding the membership complexity.
        

		To see that checking emptiness of $E$ is in \PSPACE,
		notice that if $E$ is nonempty, then it has an element 
		of size at most $\norm{E}$. We can guess such an element $C$
		in polynomial space (by  representing each coefficient in binary), and verify that $C$ is indeed in $E$ by means 
		of the \PSPACE \ membership algorithm.
\end{proof}

This result is a powerful tool which can be used to prove that a host of problems are in \PSPACE \ for RBN.
For instance, the cube-reachability problem for cubes $\cube$ and $\cube'$ is just
checking if $\poststar(\cube) \cap \cube'$ is empty, which by the  \PSPACE \ Theorem can
be done in \PSPACE. 
Combining this with Remark~\ref{remark:FSTTCS12}, we obtain the following result.
\begin{theorem}\label{thm:cube-reachability}
Cube-reachability is \PSPACE-complete for RBN.
\end{theorem}

By the reduction given in Section 4.2 of ~\cite{Gandalf21}, this result also proves
that cube-reachability is \PSPACE-complete for asynchronous shared-memory systems (ASMS), which is another model of distributed computation where agents communicate by a shared register. Due to lack of space, we
defer a discussion of this result to the appendix.

We will demonstrate further applications of the \PSPACE \ Theorem in the next  section.


%% file: almost-sure-cov.tex
Having presented our \PSPACE{} Theorem and the closure property for reachability sets of counting sets,
 we now provide two applications.
 For the first one, we consider the \emph{almost-sure coverability} problem for RBN. 
Using our new results, we  prove that this problem is \PSPACE{}-complete. 

The rest of the section is as follows: We first recall the definition of the almost-sure coverability problem, 
give a characterization of it in terms of counting sets and then prove \PSPACE{}-completeness.
Throughout this section, we fix a RBN $\RBN = (Q,\Sigma,\delta)$ with two special states $\init, \fin \in Q$, which will respectively be called the initial and final states.

\subsection{The almost-sure coverability problem}

Let $\uparrow \fin$ denote the set of all configurations $C$ of $\RBN$ such that $C(\fin) \ge 1$. For any $k \ge 1$, 
we say that the configuration $\multiset{k \cdot \init}$ \emph{almost-surely covers} $\fin$ if and only if 
$\post^*(\multiset{k \cdot \init}) \subseteq \pre^*(\uparrow \fin)$. The reason behind calling this the almost-sure coverability relation is that the definition given
here is equivalent to covering the state $\fin$ from $\multiset{k \cdot \init}$ with probability 1
under a probabilistic scheduler which picks agents uniformly at random at each step.

The number $k$ is called a \emph{cut-off} if one of the following is true: Either, 1) for all $h \ge k$, the configuration $\multiset{h \cdot \init}$ almost-surely covers $\fin$, in which case $k$ is called a \emph{positive} cut-off;
or, 2) for all $h \ge k$, the configuration $\multiset{h \cdot \init}$ does not almost-surely cover $\fin$, in which case $k$ is called a \emph{negative} cut-off.
The following was proved in Theorem 9 of~\cite{Gandalf21}.

\begin{theorem}
	Given an RBN with two states $\init,\fin$, a cut-off always exists. 
	Whether the cut-off is positive or negative can be decided in \EXPSPACE{}.
\end{theorem}

Our main result of this section is that
\begin{theorem}
	Deciding whether the cut-off of a given RBN is positive or negative is \PSPACE{}-complete. Moreover, a given RBN always has 
	a cut-off which is at most exponential in its number of states.	
\end{theorem}

\subsection{A characterization of almost-sure coverability}

We now rewrite the definition of almost-sure coverability in terms of counting sets. Let $[\init]$ be the cube such that $L(q) = U(q) = 0$ if $q \neq \init$
and $L(\init) = 0, U(\init) = \infty$. 
Notice that by definition, $\uparrow \fin$ is a cube.
We now consider the set of configurations defined by $\cSet := \poststar([\init]) \cap \overline{\prestar(\uparrow \fin)}$. 
By our \PSPACE{} Theorem \ref{thm:pspace-lemma}, $\cSet$ is a counting set such that the norm of $\cSet$ is at most $2^{p(|Q|)}$ for some fixed
polynomial $p$. We now claim the following.

\begin{theorem}~\label{thm:charac-almost-sure-cov}
	$\RBN$ has a positive cut-off if and only if $\cSet$ is finite. Moreover, $|Q| \cdot |\cSet|$ is an upper bound on the size of the cut-off for $\RBN$ and
	so $\RBN$ has a cut-off which is exponential in its number of states.
\end{theorem}

\begin{proof}
	Let $N$ be the norm of $\cSet$. Suppose $\cSet$ is finite. If $C \in \cSet$, 
	then $\sum_{q \in Q} C(q) \le |Q| \cdot N$. So, 
	if $C$ is any configuration of size $h > |Q| \cdot N$ such that $C \in \poststar(\multiset{h \cdot \init})$ then $C \in \prestar(\uparrow \fin)$. Hence, $|Q| \cdot N$ is a positive cut-off for $\RBN$.
	
	Suppose $\cSet$ is infinite, and let $\cup_i \cube_i$ be a counting constraint for $\cSet$
	whose norm is $N$. Then there must exist an index $i$ with $\cube_i := (L,U)$ and a state $p$ such that $U(p) = \infty$.
	For each $h \ge N$, consider the configuration $C_h$ given by $C_h(q) = L(q)$ if $q \neq p$ and $C_h(p) = h$. 
	Notice that $C_h \in \cSet$ and so $C_h \in \poststar([\init]) \cap \overline{\prestar(\uparrow \fin)}$.
	Hence, for every $h \ge |Q| \cdot N$, we have exhibited a configuration of size $h$, reachable from $(\multiset{h \cdot \init}$ but 
	from which $\fin$ is not coverable. Thus $N$ is a negative cut-off for $\RBN$.
\end{proof}

\begin{remark}
	Notice that we have shown that if $\cSet$ is finite, then $\RBN$ has a positive cut-off and if $\cSet$ is infinite, then $\RBN$ has a negative cut-off. This gives
	an alternative proof of the fact that a cut-off always exists for a given RBN.
\end{remark}

\subsection{\PSPACE{}-completeness of the almost-sure coverability problem}

Because of Theorem~\ref{thm:charac-almost-sure-cov}, we now have the following result.

\begin{restatable}{lemma}{LemAlmostSureUpperBound}\label{lem:almost-sure-upper-bound}
	Deciding whether the cut-off of a given RBN is positive or negative can be done in \PSPACE{}.	
\end{restatable}

\begin{proof}[Proof Sketch]
	By Theorem~\ref{thm:charac-almost-sure-cov}, it follows that a given RBN has a negative cut-off iff $\cSet = \poststar([\init]) \cap \overline{\prestar(\uparrow \fin)}$ is infinite. 
	We have already seen that $\cSet$ is a counting set such that the norm of $\cSet$ is at most $N := 2^{p(|Q|)}$ for some fixed
	polynomial $p$. 
	
	Let $\cup_i \cube_i$ be a counting constraint for $\cSet$ which minimizes its norm
	and let each $\cube_i = (L_i,U_i)$. Hence, $L_i(q) \le N$ for every state $q$. Further, 
	$\cSet$ is infinite iff there is an index $i$ and a state $q$ such that $U_i(q) = \infty$.
	Using these two facts, we can then show that $\cSet$ is infinite iff there is a state $q$ and a configuration $C \in \cSet$ such that
	$C(q') \le N$ for every $q' \neq q$ and $C(q) = N+1$. 
	
	Hence, to check if $\cSet$ is infinite, we just have to guess a state $q$ and a configuration $C$ such that $C(q') \le N$ for every $q' \neq q$ and $C(q) = N+1$ and check if $C \in \cSet$.
	Since guessing $C$ can be done in polynomial space (by representing every number in binary), by the \PSPACE{} Theorem (Theorem \ref{thm:pspace-lemma}), we can check if $C \in \cSet$ in polynomial space as well,
	which concludes the proof of the theorem.
\end{proof}

We also have the accompanying hardness result.

\begin{restatable}{lemma}{LemAlmostSureLowerBound}
	Deciding whether the cut-off of a given RBN is positive or negative is \PSPACE{}-hard.	
\end{restatable}

Similar to the cube-reachability problem,
our result on almost-sure coverability also applies to the related model of ASMS. 
This solves an open problem from~\cite{ICALPPatricia}.
For lack of space, we once again defer this discussion to the appendix.

%% file: computation.tex
In this section we give another application of our results.
We introduce a model of computation using RBN called \emph{RBN protocols}.
We take inspiration from the extensively-studied model of population protocols~\cite{First-Pop-Prot,Comp-Power-Pop-Prot,DC}.
The reader can consult the above references for more details on population protocols.




In our model, reconfigurable networks of identical, anonymous agents interact 
to compute a predicate $\varphi: \N^k \rightarrow \set{0,1}$. We show that RBN protocols compute exactly the threshold predicates, which we will define more formally below.


\subsection{RBN Protocols}

We introduce our computation model. The notation mimics that of \cite{EsparzaRW19}. 
\bala{By PN19, did you mean Petri Nets 19?}

\begin{definition}
An \emph{RBN protocol} is a tuple $\PP=(Q, \Sigma, \delta, I, O)$
where $(Q, \Sigma, \delta)$ is an RBN, $I=\set{q_1, \ldots, q_k}$ is a set of \emph{input states}, and $O:Q \rightarrow \set{0,1}$ is an \emph{output function}.
\end{definition}

\emph{Configurations} and \emph{runs} of $\PP$ are the same as that of the underlying RBN.
A configuration $C$ is called a \emph{0-consensus} (respectively a \emph{1-consensus}) if $C(q)>0$ implies $O(q)=0$ (respectively $O(q)=1$).
For $b \in \{0,1\}$, a $b$-consensus $C$ is \emph{stable} if every configuration reachable from $C$ is also a $b$-consensus.
A run $C_0 \act{} C_1 \act{} C_2 \cdots$ of $\PP$ is \emph{fair} if it is finite and cannot be extended by any step\chana{same as in PN19, stronger or same than in DC (could change to DC)}
, or if it is infinite and the following condition holds for all configurations $C,C'$: if  $C \act{} C'$ and $C=C_i$ for infinitely many $i \geq 0$, then the step $C \act{} C'$ appears
infinitely along the run.
In other words, if a fair run reaches a configuration infinitely often, then all the configurations reachable in a step from that configuration will be reached infinitely often from it.

A fair run $C_0 \act{} C_1 \act{} \dots$ \emph{converges to $b$} if there is $i \geq 0$ such that $C_{j}$ is a $b$-consensus for every $j \geq i$.
For every $\vec{v} \in \mathbb{N}^k$, 
let $\Cv$ be the configuration given by $\Cv(q_i) = \vec{v}_i$ for every $q_i \in I$, and $\Cv(q) = 0$ for every $q \in  Q \setminus I$. 
We call $\Cv$ the \emph{initial configuration for input $\vec{v}$}. 
The protocol $\PP$ \emph{computes the predicate} $\varphi \colon \mathbb{N}^k \rightarrow \{0,1\}$, if for every $\vec{v} \in \mathbb{N}^k$, every fair run starting at $\Cv$ converges to $\varphi(\vec{v})$.

\begin{figure}
\input{fig-protocol}
\caption{An RBN protocol $\PP$.}
\label{fig:protocol}
\end{figure}
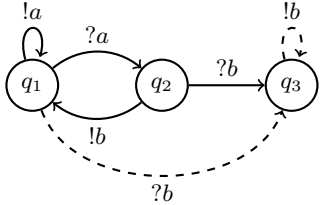

\begin{example}
Adding the dashed line transitions to the RBN of Example \ref{ex:rbn} yields the RBN protocol $\PP=(Q, \Sigma, \delta, I, O)$ illustrated in Figure \ref{fig:protocol}.
The initial state is $q_1$, i.e. $I=\set{q_1}$, 
and the  output function is defined such that $O(q_1)=O(q_2)=0$ and $O(q_3)=1$.
If there is a process in $q_3$, it can ``attract" the rest of the processes there using the new dashed transitions.
As with the  RBN of Example \ref{ex:rbn}, a process can be put in $q_3$ starting from the initial  configuration $\multiset{k \cdot q_1}$ if and only if $k\ge 3$.
This RBN protocol computes the predicate $x \ge 3$: if there are less than $3$ processes originally in $q_1$ then they stay  in states with output $0$, and if there are more, then in a fair run a process eventually enters $q_3$, and eventually the others follow, thus converging to $1$.
\end{example}


\subsection{Expressivity}

In this section, we show that RBN protocols compute exactly the predicates definable by counting sets. A predicate $\varphi:\N^k \rightarrow \set{0,1}$ is \emph{definable by counting sets} if for every $b \in \{0,1\}$, the sets $\set{\vec{v} \ | \ \varphi(\vec{v})=b}$ are counting sets.

For $b \in \set{0,1}$, define the following sets of configurations:
\begin{itemize}
\item Let $\cons_b$ be the set of $b$-consensus configurations.
\item Let $\stable_b$ be the set $\overline{\pre^*\left(\overline{\cons_b}\right)}$ of stable $b$-consensuses. These are the configurations from which one can reach only $b$-consensuses.
\item Let $\initcube_b$ be the set of initial configurations $\Cv$ for inputs $\vec{v}$ such that  $\varphi(\vec{v})=b$.
\end{itemize}

The next lemma states that every predicate computed by a protocol is definable by counting sets.

\begin{restatable}{lemma}{LemAtMostCC}\label{lm:at-most-cc}
Let $\PP$ be a RBN protocol that computes the predicate $\varphi:\N^k \rightarrow \set{0,1}$.
Then for every $b \in \{0,1\}$, the sets $\initcube_b, \cons_b$ and $\stable_b$ are all counting sets.
This entails that $\varphi$ is definable by counting sets.
\end{restatable}

\begin{proof}[Proof Sketch]
	Fix a $b \in \{0,1\}$. It is easy to see that $\cons_b$ is a cube.
	Unraveling the definitions of $\initcube_b$ and $\stable_b$, we can express them
	in terms of $\cons_b$ by using boolean operations and $\prestar$. By the Closure Corollary
	(Corollary~\ref{coro:closure}), 
	they are counting sets. Set $\set{\vec{v} \ | \ \varphi(\vec{v})=b}$ is simply $\initcube_b$ restricted  to $I$,  and so we are done.
\end{proof}

The next lemma states the converse result. It essentially uses the fact that there is a 
sub-class of population protocols called IO protocols which compute exactly the predicates definable by  counting sets (Theorem 7 and Theorem 39 of~\cite{Comp-Power-Pop-Prot,EsparzaRW19}), and
that IO protocols are a sub-class of RBN (Section 6.2 of~\cite{Gandalf21}).


\begin{restatable}{lemma}{LemAtLeastCC}\label{lm:at-least-cc}
	Let $\varphi:\N^k \rightarrow \set{0,1}$ be a predicate definable by counting sets.
	Then there exists a RBN protocol computing $\varphi$.
\end{restatable}

By Lemma \ref{lm:at-most-cc} and Lemma \ref{lm:at-least-cc}, 
we get our result.

\begin{theorem}
\label{thm:rbn-compute}
RBN protocols compute exactly the predicates definable by counting sets.
\end{theorem}

%
%

%% file: fig-protocol.tex
\begin{center}
    \begin{tikzpicture}[->, thick]
      \node[place] (a1) {$q_1$};
	  \node[place] (a2) [right =of a1] {$q_2$};
	  \node[place] (a3) [right =of a2] {$q_3$};

      \path[->]
      (a1) edge[bend left = 40] node[above] {$?a$} (a2)
      (a2) edge[bend left = 40] node[below] {$!b$} (a1)
      (a2) edge node[above] {$?b$} (a3)
      (a1) edge[loop above] node[above] {$!a$} (a1)
      ;
      
      \path[->,dashed]
      (a1) edge[bend right = 70] node[below] {$?b$} (a3)
      (a3) edge[loop above] node[above] {$!b$} (a3)
      ;
    \end{tikzpicture}
\end{center}

%% file: appendix.tex
\section{Proofs for section~\ref{sec:main-sec}}

\LemmaRuns*

\begin{proof} \bala{Can you check if this proof is okay and elaborate if needed?}\chana{done}
Let $C \trans{t+t_1,\dots,t_n} C'$, 
where $t = (p,!a,p')$ and each $t_i = (p_i,?a,p_i')$. 
Let $\theta=(v,S)$ such that $C \in \supp{\theta}$.
Let $t_1, \ldots, t_k$
be the subset of receive transitions 
such that $p_1,\ldots,p_k \notin S$, 
and let $t_{k+1}, \ldots, t_n$
be the subset of receive transitions 
such that $p_{k+1},\ldots,p_n \in S$.	
We have two cases. 
\begin{itemize}
\item  If $p \notin S$, then we  perform a ``broadcast from $v$".
Let $v' = v - \sum_i \vec{p_i} + \sum_{i=1}^k \vec{p_i'} - \vec{p} + \vec{p'}$.
Let $S'= S \cup \set{ p_i' |  i  \in \set{k+1,n} }$.
\item  If $p \in S$, then we  perform a ``broadcast from $S$".
Let $v' = v - \sum_i \vec{p_i} + \sum_{i=k+1}^n \vec{p_i'}$.
Let $S'= S \cup \set{p'} \cup  \set{ p_i' |  i  \in \set{k+1,n} }$.
\end{itemize}
Since $C \trans{t+t_1,\dots,t_n} C'$,
 $C(q)>0$ for $q \in \set{p,p_1,\ldots,p_n}$ and thus $v'$ is well-defined.
Let $\theta'$ be $(v', S')$. 
By our definition of $v', S'$, there is an edge 
$\theta \rightsquigarrow^a \theta'$ in the symbolic graph.
\end{proof}

\LemmaNF*

\begin{proof}
	
	Let $\theta = \theta_0 \actsymb{} \theta_1 \actsymb{} \theta_2 \actsymb{} \dots \theta_{m-1} \actsymb{} \theta_m = \theta'$ be the path between $\theta$ and $\theta'$.
	We proceed by induction on $m$. The claim is clearly true for $m = 0$. Suppose $m > 0$ and the claim is true for $m-1$. By induction hypothesis, we can assume 
	that the path $\theta_0 \actsymb{} \theta_1 \actsymb{} \dots \actsymb{} \theta_{m-1}$ is already in normal form.
	
	Let each $\theta_i = (v_i,S_i)$. Let $l$ be the number of bad pairs in the path between $\theta_0$ and $\theta_m$. If $l = 0$, then the path is already in normal form and we are done.
	Suppose $l > 0$ and let $(w,w')$ be a bad pair. Since the path between $\theta_0$ and $\theta_{m-1}$ is already in normal form, it has to be the case that $w' = m$.
	Hence, we have $Z := (S_{w} \setminus S_{w+1}) \cap S_m \neq \emptyset$.
	
	By Proposition~\ref{prop-monotone}, the following is a valid path: $(v_{w},S_{w}) \actsymb{} (v_{w+1},S_{w+1} \cup Z) \actsymb{} (v_{w+2},S_{w+2} \cup Z) \dots (v_{m-1}, S_{m-1} \cup Z)\actsymb{} (v_m, S_m \cup Z) = (v_m,S_m)$. 
	Let $\theta'_j := \theta_j$ if $j \le w$ and $(v_j,S_j \cup Z)$ otherwise. 
	Hence, we get a path $\theta'_0 \actsymb{} \theta'_1 \actsymb{} \dots \theta'_{m-1} \actsymb{} \theta'_m$.
	
	Let each $\theta'_i = (v_i',S_i')$. We first claim that the path between $\theta'_0$ and $\theta'_{m-1}$ is in normal form. 
	Indeed, suppose there exists $0 \le i < j \le m-1$ such that $(S_i' \setminus S_{i+1}') \cap S_j' \neq \emptyset$. There are four cases:
	\begin{itemize}
		\item $i \le w$ and $j \le w$ : In this case $S_i' \setminus S_{i+1}' = S_i \setminus S_{i+1}$ and $S_j' = S_j$, and since the path between $\theta_0$ and $\theta_{m-1}$ is in normal form, this case cannot happen.
		\item $w < i$ and $w < j$: In this case $S_i' \setminus S_{i+1}' = S_i \setminus S_{i+1}$ and $S_j' = S_j \cup Z$. Since the path between $\theta_0$ and $\theta_{m-1}$ is in normal form, this should then 
		imply that $(S_i \setminus S_{i+1}) \cap Z \neq \emptyset$. By definition this means that $(S_w \setminus S_{w+1}) \cap S_i \neq \emptyset$ which contradicts the fact
		that the path between $\theta_0$ and $\theta_{m-1}$ is in normal form.
		\item $i < w$ and $w < j$: Similar to the case before, this should then imply that $(S_i \setminus S_{i+1}) \cap Z \neq \emptyset$. 
		By definition this means that $(S_i \setminus S_{i+1}) \cap S_w \neq \emptyset$ which contradicts the fact
		that the path between $\theta_0$ and $\theta_{m-1}$ is in normal form.
		\item $i = w$ and $w < j$: In this case $S_i' \setminus S_{i+1}' = S_i \setminus (S_{i+1} \cup Z)$ and $S_j' = S_j \cup Z$.  This
		would then imply that $(S_i \setminus S_{i+1}) \cap S_j \neq \emptyset$ which contradicts the fact
		that the path between $\theta_0$ and $\theta_{m-1}$ is in normal form.
	\end{itemize}
	
	It is then clear by construction, that this new path from $\theta'_0 := \theta_0$ to $\theta'_m := \theta_m$ has at most $l-1$ bad pairs only.
	Hence, we now have a path from $\theta_0$ to $\theta_m$ such that the prefix of length $m-1$ is in normal form and the number of bad pairs has been strictly reduced to $l-1$.
	Repeatedly applying this procedure, leads to a path in normal form between $\theta_0$ and $\theta_m$.
\end{proof}

\TheoremUpperBoundN*

\begin{proof}
	Suppose $\theta \actsymb{*} \theta'$. If the length of the path is 0, then there is nothing to prove.
	Hence, we restrict ourselves to the case when the length of the path is bigger than 0.
	By Lemma~\ref{lemma-nf}, there is a path in normal from from $\theta$ to $\theta'$ (say)
	$\theta = \theta_0 \actsymb{} \theta_1 \actsymb{} \theta_2 \dots \theta_{m-1} \actsymb{} \theta_m = \theta'$ with each
	$\theta_i := (v_i,S_i)$. 
	
	Let $N_0 = 0$ and let $N_i = (N_{i-1}+1) \cdot (|S_{i-1} \setminus S_i|+1)$ for every $1 \le i \le m$. 
	In Lemma 6 of~\cite{Delzanno12}, the following fact has been proved: 
	\begin{quote}
		For every $1 \le i \le m$ and for every $C' \in \supp{\theta_i}_{N_i+1}$, there exists $C \in \supp{\theta_{i-1}}_{N_{i-1}+1}$ such that
		$C \act{*} C'$.
	\end{quote}
	
	This immediately proves that for all $C' \in \supp{\theta'}_{N_m + 1}$, there exists $C \in \supp{\theta}$ such that $C \act{*} C'$. If we prove $N_m \le k \times (2k)^{|Q|} \times (|Q|+1)^{|Q|+1}$, then the proof of the theorem will be complete.
	
	Notice that since the path between $\theta_0$ and $\theta_m$ is in normal form, the number of indices $i$ such that $|S_{i-1} \setminus S_i| > 0$
	is at most $|Q|$. Indeed, suppose $q \in S_{i-1} \setminus S_i$ for some $i$. Then by the normal form property, $q \notin S_j$ for any $j \ge i$. 
	Hence, in the rest of the path $q$ does not appear in the abstract part at all. By definition of the edges in the symbolic graph, if $(v,\emptyset) \actsymb{} (v',S')$ is an edge,
	then $S' = \emptyset$. These two facts then imply that
	the number of indices $i$ such that $|S_{i-1} \setminus S_i| > 0$
	is at most $|Q|$.
	
	It then follows that except for at most $|Q|$ indices, each index $N_i$ is obtained from $N_{i-1}$ by simply adding 1
	and in the remaining indices, $N_i$ is obtained from $N_{i-1}$ by adding 1 and then multiplying by a number which is at most $|Q|+1$.
	The way to maximize $N_m$ by this procedure is by letting $N_i = N_{i-1} + 1$ for every $1 \le i \le m-|Q|$ and then
	letting $N_i = (N_{i-1}+1) \cdot (|Q|+1)$ for every $m-|Q| < i \le m$. This gives an upper bound 
	of $(m-|Q|+1)|Q|(|Q|+1)^{|Q|}$ for $N_m$. Since $m$ is itself the length of the path between $\theta_0$ and $\theta_m$,
	$m$ is upper bounded by the number of symbolic configurations in $\gr_k$ which is at most $|\nn_k^{Q}| \times 2^{|Q|} \le k \times k^{|Q|} \times 2^{|Q|}$.
	Overall we get that $N_m \le k \times (2k)^{|Q|} \times (|Q|+1)^{|Q|+1}$.
\end{proof}

\section{Proofs for Section~\ref{sec:almost-sure}}

\LemAlmostSureUpperBound*

\begin{proof}
	By Theorem~\ref{thm:charac-almost-sure-cov}, it follows that a given RBN has a negative cut-off iff $\cSet = \poststar([\init]) \cap \overline{\prestar(\uparrow \fin)}$ is infinite. 
	We have already seen that $\cSet$ is a counting set such that the norm of $\cSet$ is at most $N := 2^{p(|Q|)}$ for some fixed
	polynomial $p$. 
	
	Let $\cup_i \cube_i$ be a counting constraint for $\cSet$ which minimizes its norm
	and let each $\cube_i = (L_i,U_i)$. Hence, $L_i(q) \le N$ for every state $q$. Further, 
	$\cSet$ is infinite iff there is an index $i$ and a state $q$ such that $U_i(q) = \infty$.
	
	Using these two facts, we claim that $\cSet$ is infinite iff there is a state $q$ and a configuration $C \in \cSet$ such that
	$C(q') \le N$ for every $q' \neq q$ and $C(q) = N+1$. Indeed, if $\cSet$ is infinite, then there is an $i$ and a $q$ such that $U_i(q) = \infty$. 
	If we let $C$ be such that $C(q') = L(q') \le N$ for every $q' \neq q$
	and $C(q) = N+1$ then $C \in \cube_i \in \cSet$. 
	
	For the other direction, suppose such a state $q$ and a configuration $C$ exists. Since $C \in \cSet$ we have that $C \in \cube_i$ for 
	some $i$. Now, since $C(q) = N+1$ and the norm of $\cube_i$ is at most $N$,
	it must be the case that $U_i(q) = \infty$. This then proves that $\cSet$ is infinite.
	
	Hence, to check if $\cSet$ is infinite, we just have to guess a state $q$ and a configuration $C$ such that $C(q') \le N$ for every $q' \neq q$ and $C(q) = N+1$ and check if $C \in \cSet$.
	Since guessing $C$ can be done in polynomial space (by representing every number in binary), by the \PSPACE{} Theorem (Theorem \ref{thm:pspace-lemma}), we can check if $C \in \cSet$ in polynomial space as well,
	which concludes the proof of the theorem.
\end{proof}

\LemAlmostSureLowerBound*

\begin{proof}

	We reduce from the fixed-configuration almost-sure coverability problem for RBN. 
	In this problem, we are given a RBN $\RBN = (Q,\Sigma,\delta)$, a configuration $C$ of $\RBN$
	such that $2 \le |C| \le |Q|$ and a state $q_f \in Q$ and we are asked to decide if 
	$C$ can almost-surely cover $q_f$, i.e., if $\poststar(C) \subseteq \prestar(\uparrow q_f)$.
	This problem is \PSPACE{}-hard and the proof is as follows: In Theorem 4 of~\cite{EsparzaRW19},
	the authors give a reduction from the acceptance problem for linear-space bounded Turing machines
	to the problem of covering a state $q_f$ from a given initial configuration $C$ for a subclass of RBN called IO nets which have the following property:
	Starting from the initial configuration of the IO net, there is exactly one 
	execution which is possible. It then follows that the covering the state $q_f$ from $C$
	is equivalent to almost-surely covering $q_f$ from $C$. Since IO nets are a subclass of 
	RBN (Section 6.2 of~\cite{Gandalf21}), it follows that the fixed-configuration
	almost-sure coverability problem for RBN is \PSPACE-hard.

	We now give a reduction from the fixed-configuration almost-sure coverability problem for RBN to the 
	problem of checking if a given RBN has a positive cut-off.
	Let $(\RBN, C, q_f)$ be an instance of the fixed-configuration almost-sure coverability problem for RBN such that $\RBN = (Q,\Sigma,\delta)$ and $C = \multiset{q_1,\dots,q_n}$. 
	The required reduction proceeds in three stages. \\
	
	\textbf{First stage: } We construct a new RBN $\RBN_1 = (Q_1,\Sigma_1,\delta_1)$ as follows: $Q_1 = Q \times \{1,2,\dots,n\} \cup \{\fin\}$ where
	$\fin$ is a new state, $\Sigma_1 = \Sigma \cup \{@\}$ where $@$ is a new letter 
	and $\delta_1 = \{(p,i) \act{!a} (q,i) : p \act{!a} q \in \delta\} \cup \{(p,i) \act{?a} (q,i) : p \act{?a} q \in \delta\} \cup 
	\{(q_f,i) \act{!@} \fin\}$. For each $i$, the set $Q \times \{i\}$, will be called the $i^{th}$ copy of $\RBN$.
	
	Intuitively, $\RBN_1$ contains $n$ copies of $\RBN$ along with a new state $\fin$ such that it is always possible to move
	from any copy of the state $q_f$ to the new state $\fin$. Note that since $n = |C| \le |Q|$, this construction takes polynomial time.
	
	Let $D := \{(q_1,1),(q_2,2),\dots,(q_n,n)\}$. It is straightforward to verify that 
	$C$ can almost-surely cover $q_f$ in $\RBN$ iff
	$D$ can almost-surely cover $\fin$ in $\RBN_1$.\\
	
	\textbf{Second stage: } We now construct a second RBN $\RBN_2 = (Q_2, \Sigma_2, \delta_2)$ as follows: 
	$Q_2 = Q_1 \cup \{\init\}$ where 
	$\init$ is a new state, 
	$\Sigma_2 = \Sigma_1 \cup \{\#\} \cup \{\$_1,\dots,\$_n\}$ where $\#,\$_1,\dots,\$_n$ are $n+1$ new letters and 
	$\delta_2$ contains all the transitions in $\delta_1$ and also the following transitions:
	\begin{itemize}
		\item Type 1 transitions: $\init \act{!\#} (q_i,i)$ and for each $1 \le i \le n$.
		
		\item Type 2 transitions: For every $p \in Q$ and $1 \le i \le n$, we have the transitions 
		$(p,i) \act{!\$_i} (p,i)$ and $(p,i) \act{?\$_i} \fin$.
	\end{itemize}
	
	By combining the Type 1 and Type 2 transitions, it is very easy to verify the following facts:
	\begin{itemize}
		\item Fact 1: If $C' \ge \multiset{(p,i),(q,i)}$ for some $i$ and some $p,q$, then $C'$ can cover $\fin$.
		\item Fact 2: If $C' \ge \multiset{2 \cdot \init}$ or $C \ge \multiset{\init, (p,i)}$ for some $i$ and some $p$, then $C'$ can cover $\fin$.
		\item Fact 3: $n+1$ is a positive cut-off for $\RBN_2$.
		\item Fact 4: If $C' \act[\RBN_2]{*} C''$ is a run such that $C'(\init) = 0$ and $C'$ does not contain two processes in the same
		copy of $\RBN$, then no transitions of type 1 or type 2 could have been fired along this run. Consequently, we have
		$C' \act[\RBN_1]{*} C''$. 
	\end{itemize}
	
	We now claim that 
	\begin{quote}
		$\multiset{n \cdot \init}$ can almost-surely cover $\fin$ in $\RBN_2$ iff $D$ can almost-surely cover $\fin$ in $\RBN_1$.
	\end{quote}
	
	Suppose $\multiset{n \cdot \init}$ can almost-surely cover $\fin$ in $\RBN_2$. We want to show that $D$ can almost-surely cover $\fin$ in $\RBN_1$.
	To do this, we have to show that if $D \act[\RBN_1]{*} C'$, then $C'$ can cover $\fin$ in $\RBN_1$.
	Notice that $C'(\init) = 0$ and $C'$ does not contain two processes in the same copy of $\RBN$.
	
	Notice that, by using the type 1 transitions, we have $\multiset{n \cdot \init} \act[\RBN_2]{*} D$ and so we have $\multiset{n \cdot \init} \act[\RBN_2]{*} D \act[\RBN_2]{*} C'$. By assumption, this means that $C' \act[\RBN_2]{*} C''$ with $C''(\fin) > 0$.
	By Fact 4, we have $C' \act[\RBN_1]{*} C''$ and so $C'$ can cover $\fin$ in $\RBN_1$.
	
	Suppose $D$ can almost-surely cover $\fin$ in $\RBN_1$. We want to show that $\multiset{n \cdot \init}$ can almost-surely cover $\fin$ in $\RBN_2$.
	To do so, we have to show that if $\multiset{n \cdot \init} \act[\RBN_2]{*} C'$, then $C'$ can cover $\fin$ in $\RBN_2$.
	By means of Fact 1 and Fact 2, it suffices to show that this is the case when $C'$ contains exactly one process in each copy of $\RBN$.
	In this case, we will prove that $D \act[\RBN_1]{*} C'$ and so by assumption, this means that $C'$ can cover $\fin$ in $\RBN_1$ and 
	hence in $\RBN_2$ as well.
	
	All that remains to show that is that $D \act[\RBN_1]{*} C'$, which is what we do now. Consider the run $\multiset{n \cdot \init} \act[\RBN_2]{} C_1 \act[\RBN_2]{} C_2 \dots C_k \act[\RBN_2]{} C'$.
	Since $C'$ has exactly one process in each copy of $\RBN$, it must be the case that along this run, no type 2 transitions were fired,
	and each type 1 transition was fired exactly once, i.e., for each $1 \le i \le n$, the transition $\init \act{!\#} (q_i,i)$ occured exactly once
	along this run. Notice that if for some $j$, we have $C_j \act[\RBN_2]{r_j + r^1_j, \dots, r^{l_j}_j} C_{j+1} \act[\RBN_2]{r_{j+1} + r^1_{j+1},\dots,r^{l_{j+1}}_{j+1}} C_{j+2}$ where $r_j$ is not a type 1 transition and $r_{j+1}$ is a type 1 transition, then 
	$C_j \act[\RBN_2] {r_{j+1} + r^1_{j+1},\dots,r^{l_{j+1}}_{j+1}} C'' \act[\RBN_2]{r_j + r^1_j, \dots, r^{l_j}_j} C_{j+2}$. 
	This means that we can push all the occurences of type 1 transitions along this run to the beginning. But then notice that after the first $n$ steps
	we would have reached the configuration $D$ from $\multiset{n \cdot \init}$. This means that $D \act[\RBN_2]{*} C'$
	and by Fact 4, we have $D \act[\RBN_1]{*} C'$, which finishes the proof.
	
	Notice that by Fact 3, we have actually shown the following
	\begin{quote}
		Fact 5: $D$ can almost-surely cover $\fin$ in $\RBN_1$ iff $\multiset{n \cdot \init}$ can almost-surely cover $\fin$ in $\RBN_2$ iff
		$n$ is a positive cut-off for $\RBN_2$.
	\end{quote}

	\textbf{Third stage: } We now construct our final RBN $\RBN_3 = (Q_3,\Sigma_3,\delta_3)$ as follows: $Q_3 = Q_2 \cup \{s_1,s_2,\dots,s_n\}$ where
	$s_1,\dots,s_n$ are $n$ new states, 
	$\Sigma_3 = \Sigma_2 \cup \{a_1,\dots,a_n,b\}$ where $a_1,\dots,a_n,b$ are $n+1$ new letters and 
	$\delta_3$ contains all the transitions in $\delta_2$ and also the following transitions:
	
	\begin{itemize}
		\item Type 3 transitions: For each $i \in \{1,\dots,n\}$, we have $(q_i,i) \act{!a_i} (q_i,i)$ and $s_{i-1} \act{?a_i} s_i$. (Here and in the sequel, $s_0$ is taken to be $\init$).
		\item Type 4 transitions: For each $i \in \{1,\dots,n-1\}$, we have $s_i \act{!b} \init$.
	\end{itemize}
	
	Since $s_n$ is a sink state which does not broadcast anything and since the only way to reach $s_n$ is through $s_{n-1}$, it is easy to verify the following:
	\begin{itemize}
		\item Fact 6: Suppose $C' \act[\RBN_3]{*} C''$ such that $C'(s_n) = i_n$ and $C''(s_n) = j_n$.
		Then $i_n \le j_n$ and $C' - i_n \cdot s_n  \act[\RBN_3]{*} C'' - j_n \cdot s_n + (j_n - i_n) \cdot s_{n-1} \act[\RBN_3]{*} 
		C'' - j_n \cdot s_n + (j_n - i_n) \cdot \init$.
	\end{itemize}
	
	We now claim that
	\begin{quote}
		There is a positive cut-off for $\RBN_3$ iff $n$ is a positive cut-off for $\RBN_2$.
	\end{quote}
	
	Suppose $n$ is a positive cut-off for $\RBN_2$. We claim that $n$ is also a positive cut-off for $\RBN_3$. To show this,
	we have to prove that for all $h \ge n$, if $\multiset{h \cdot \init} \act[\RBN_3]{*} C'$, then $C'$ can cover $\fin$ in $\RBN_3$.
	Let $C'(s_l) = i_l$ for every $1 \le l \le n$ and let $i = \sum_{1 \le l \le n} i_l$.
	We consider two cases:
	\begin{itemize}
		\item Case 1: $C'(s_n) = 0$. Let $\tilde{C} = C' - (\sum_{1 \le l \le n} i_l \cdot s_l) + i \cdot \init$. Notice that $\tilde{C}$ can be reached from $\multiset{h \cdot \init}$
		in $\RBN_2$ - Simply use the same run from $\multiset{h \cdot \init}$ to $C'$,
		but remove all the type 3 and type 4 transitions. By assumption then, $\tilde{C}$ can 
		cover $\fin$ in $\RBN_2$ and so in $\RBN_3$ as well.
		
		Further, notice that $C'$ can reach $\tilde{C}$ in $\RBN_3$ means of type 4 transitions.
		This means that $C'$ can also cover $\fin$ in $\RBN_3$.
		
		\item Case 2: $C'(s_n) > 0$. Let $h' = h - i_n$. By Fact 6, we have $\multiset{h' \cdot \init} \act[\RBN_3]{*} C' - i_n \cdot s_n$.
		If we show that $h' \ge n$, then we can apply the same argument as Case 1 to finish this case as well.
		Indeed, for a process to reach the state $s_n$, it is easy to see by construction that, in $C'$, there should be at least 
		one process in each copy of $\RBN$. Hence, there are at least $n$ processes in $C'$ which are not in $s_n$ and so $h' \ge n$.
	\end{itemize}
	
	Suppose there is a positive cut-off for $\RBN_3$. We need to show that $n$ is a positive cut-off for $\RBN_2$.
	By Fact 5, it suffices to show that $D$ can almost-surely cover $\fin$ in $\RBN_1$. To show this, we need
	to show that if $D \act[\RBN_1]{*} C'$ then $C'$ can cover  $\fin$ in $\RBN_1$. 
	
	Let $N$ be the positive cut-off for $\RBN_3$ and let $h \ge \max\{n,N\}$. By assumption $\multiset{h \cdot \init}$ can almost-surely
	cover $\fin$ in $\RBN_3$. By using Type 1 and Type 3 transitions, it is easy to see that, $\multiset{h \cdot \init} \act[\RBN_3]{*} 
	D + (h-n) \cdot s_n$ and so $\multiset{h \cdot \init} \act[\RBN_3]{*} C' + (h-n) \cdot s_n$.
	By assumption, $C' + (h-n) \cdot s_n$ can cover $\fin$ in $\RBN_3$ and so we have a run $C' + (h-n) \cdot s_n \act[\RBN_3]{*} C''$ with $C''(\fin) > 0$.
	By induction on the run, it is easy to prove that along this run, all the configurations $\tilde{C''}$ satisfy 
	$\tilde{C''}(s_n) = h-n$ and $\tilde{C''}(s_i) = 0$ for every $0 \le i \le n-1$.
	Hence, by Fact 6, we have $C' \act[\RBN_3]{*} C'' - (h-n) \cdot s_n$. Notice that along this run, there is no possibility of
	firing any transition of type 1, 2 or 4. Further, if a transition of type 3, i.e., a transition of the form $(q_i,i) \act{!a_i} (q_i,i)$ is fired, 
	then there could have been no process which received that message. It follows that transitions of type 3 do not change
	the configuration along this run. Hence, we can assume that no transitions belonging to type 3 are fired along this run.
	This then implies that $C' \act[\RBN_1]{*} C'' - (h-n) \cdot s_n$ and so $C'$ can cover $\fin$ in $\RBN_1$.

	This chain of constructions then proves the desired result.
\end{proof}

\section{Proof for section~\ref{sec:computation}}

%

\LemAtMostCC*

\begin{proof}
	Let $\PP=(Q, \Sigma, \delta, I, O)$ be a RBN protocol computing $\varphi$ and let $b \in \{0,1\}$.
	First we show that $\stable_b$ is a counting set.
	The set $\cons_b$ is equal to  the cube such that there are 0 processes in states $q$ with $O(q)= 1-b$ (i.e. an upper and a lower bound of $0$), 
	and an arbitrary number of processes elsewhere (i.e. an upper bound of $\infty$ and a lower bound of $0$). 
	By the Closure Corollary \ref{coro:closure}, 
	$\stable_b=\overline{\pre^*\left(\overline{\cons_b}\right)}$ is a counting set.
	
	Let $\initcube$ be the counting set of initial configurations defined by
	the cube which puts an arbitrary number of processes in initial states of $I$, 
	and $0$ elsewhere.
	The set $\initcube \cap \overline{\pre^*(\overline{\pre^*(\stable_b)})}$ is the set of initial configurations from which all runs of $\PP$ converge to $b$.  
	By the Closure Corollary \ref{coro:closure}, it is a counting set. 
	Since $\PP$ computes $\varphi$, by definition $\initcube_b=\initcube \cap \overline{\pre^*(\overline{\pre^*(\stable_b)})}$.
	The set $\set{\vec{v} \ | \ \varphi(\vec{v})=b}$ is equal to $\initcube_b$ restricted  to the initial states  $I$,  and so we are done.
\end{proof}

\LemAtLeastCC*

\begin{proof}
	To prove this, we first need the notion of an immediate observation net.
	An immediate observation (IO) net is a tuple $\net = (Q,\delta)$ 
	where $Q$ is a finite set of states 
	and $\delta \subseteq Q \times Q \times Q$ is the transition relation.
	A configuration of $\net$ is a multiset over $Q$, and there is a step between two configurations $C,C'$ if there exists 
	$(p,q,p') \in \delta$
	such that $C \ge \multiset{p,q}$ and $C' = C - \vec{p} + \vec{p'}$. 
	
	In Section 6.2 of~\cite{Gandalf21}, it is shown that RBN can simulate IO nets. 
	More specifically, given an IO net $\net = (Q,\delta)$, Section 6.2 of \cite{Gandalf21} shows that we can compute in polynomial time, a RBN $\RBN=(Q, \Sigma, \delta')$ with the same set of states
	such that $C \act{*} C'$ in $\net$ if and only if $C \act{*} C'$ in $\RBN$ for any two configurations $C,C'$.
	This implies that for any two subsets of configurations $\cube, \cube'$, $\poststar(\cube) \subseteq \prestar(\cube')$ in $\net$ if and only if
	$\poststar(\cube) \subseteq \prestar(\cube')$ in $\RBN$.
	
	Given an IO net $\net = (Q,\delta)$, a subset $I \subseteq Q$ and an output function $O : Q \to \{0,1\}$, 
	the tuple $(\net, I, O)$ defines an immediate observation (IO) population protocol, a subclass of population protocols introduced in~\cite{Comp-Power-Pop-Prot}. 
	Similar to the definition of RBN protocols, we can define the notion of an IO protocol computing a predicate and 
	Theorem 7 and Theorem 39 of~\cite{Comp-Power-Pop-Prot} shows that IO population protocols compute exactly the predicates definable by counting sets.
	
	Proposition 2.12 of~\cite{DC} entails that an IO protocol computes a predicate $\varphi$ if and only if $\poststar(\initcube_b) \subseteq \prestar(\stable_b)$
	for every $b \in \{0,1\}$. 
This is also true for RBN. 
Indeed Proposition 2.12 of~\cite{DC} states the above result for “well-behaved generalized
	protocols” (Definition 2.1 of~\cite{DC}). Fix an arbitrary RBN protocol $\prot$. 
	By definition 2.1 of~\cite{DC}, it is a generalized
	protocol by setting $Conf$ to be the set of configurations of $\prot$, $\Sigma$ to
	be its set of states, and $Step$ to be the step relation of the underlying RBN. 
	By definition 2.8 of~\cite{DC}, it is well-behaved, i.e.,
	every fair execution eventually ends up in a bottom strongly connected component of the reachability graph. 
	This is because the number of processes does not change along a run, so the reachability graph
	from any configuration is finite. 
	
	Since RBN can simulate IO nets, it follows 
	that RBN protocols can compute any predicate computable by IO protocols and this concludes
	the proof.
	
\end{proof}

%
%

\section{Asynchronous shared-memory systems}

We now consider another model of distributed computation called
asynhcronous shared-memory systems (ASMS)~\cite{Hague11,JACM16}. Here, we have a set 
of finite-state, anonymous agents which can communicate by means of a single shared register,
i.e., agents can either write a value to the register or read the value currently written on the
register. The definitions and notations in this section are taken from~\cite{Gandalf21}.

\begin{definition}
	An asynchronous shared-memory system (ASMS) is a tuple $\ASMS = (Q,\Sigma,\delta)$ where
	$Q$ is a finite set of states, $\Sigma$ is a finite alphabet,
	and $\delta \subseteq Q \times \{R,W\} \times \Sigma \times Q$ is the set of transitions.
	Here $R$ stands for \emph{read}, and
	$W$ stands for \emph{write}.
\end{definition}

We use $p \trans{R(d)} q$ (resp. $p \trans{W(d)} q$) to denote
that $(p,R,d,q) \in \delta$ (resp. $(p,W,d,q) \in \delta$). 
A configuration $C$ of an ASMS is a multiset over $Q \cup \Sigma$
such that $\sum_{d \in \Sigma} C(d) = 1$, i.e., $C$ contains
exactly one element from the set $\Sigma$. 
Hence, we sometimes denote a configuration $C$ as $(M,d)$
where $M$ is a multiset over $Q$ (which counts the number of processes
in each state) and $d \in \Sigma$ (which denotes the content of the shared register).
The value $d$ will be denoted by $\data(C)$.

A \emph{step} between configurations $C = (M,d)$ and $C' = (M',d')$ exists
if there is $t = (p,\oper,d'',q) \in \delta$ such that
$M(p) > 0$, $M' = M - \vec{p} + \vec{q}$
and either $\oper = R$ and $d = d' = d''$ or 
$\oper = W$ and $d' = d''$. If such a step exists,
we denote it by $C \trans{t} C'$ and we let $\trans{*}$ denote
the reflexive transitive closure of the step relation. 
A \emph{run} is then a sequence of steps.

A cube $\cube = (L,U)$ of an ASMS $\ASMS = (Q,\Sigma,\delta)$ is defined to be a cube over $Q \cup \Sigma$
satisfying the following property : There exists $d \in \Sigma$ such that $L(d) = U(d) = 1$ and $L(d') = U(d') = 0$ for every other $d' \in \Sigma$.
Hence, we sometimes denote a cube $\cube$ as $(L,U,d)$ where $(L,U)$ is a cube over $Q$ and $d \in \Sigma$.
Membership of a configuration $C$ in a cube $\cube$ is then defined in a straightforward manner.

The cube-reachability problem for ASMS is then to decide, given $\ASMS$ and two cubes $\cube, \cube'$,
whether $\cube$ can reach $\cube'$, i.e., whether there are configurations $C \in \cube, C' \in \cube'$
such that $C \act{*} C'$. In Section 4 of~\cite{Gandalf21}, it was shown that the
cube-reachability problems for RBN and ASMS are polynomial-time equivalent. By Theorem~\ref{thm:cube-reachability}, we get

\begin{theorem}
	The cube-reachability problem for ASMS is \PSPACE-complete.
\end{theorem}

\subsection*{The almost-sure coverability problem}

Similar to RBN, we can define the almost-sure coverability problem for ASMS.
Let $\ASMS = (Q,\Sigma,\delta)$ be an ASMS with two special states $\init$ and $\fin$
and a special initial letter $\# \in \Sigma$.
Let $\uparrow \fin$ denote the set of all configurations $C$ with $C(\fin) \ge 1$.
For any $k \ge 1$, we say that the configuration $(\multiset{k \cdot \init},\#)$ almost-surely covers 
$\fin$ iff $\poststar(\multiset{k \cdot \init},\#) \subseteq \prestar(\uparrow \fin)$.
Now, similar to RBN, it is easy to define the notion of a cut-off for ASMS. The following
fact is known (Theorem 3 of~\cite{ICALPPatricia}).

\begin{theorem}
	Given an ASMS $\ASMS$ with two state $\init, \fin$ and a letter $\#$, a cut-off always exists.
	Whether the cut-off is positive or negative can be decided in \EXPSPACE \ and \PSPACE-hard.
\end{theorem} 

The main result of this subsection is that
\begin{theorem}
	Deciding whether the cut-off of a given ASMS is positive or negative is in \PSPACE-complete.
\end{theorem}

Note that it suffices only to prove the upper bound, since the lower bound is already known.
Let $\ASMS = (Q,\Sigma,\delta)$ be a fixed ASMS with $\init, \fin \in Q$ and $\# \in \Sigma$.
Let $([\init],\#)$ denote the cube which has an arbitrary number of agents in the state $\init$
and 0 elsewhere. Similar to the model of RBN, we first show that, 
\begin{theorem}
	$\ASMS$ has a positive cut-off iff $\cSet := \poststar(([\init],\#)) \cap \overline{\prestar(\uparrow \fin)}$ is finite. 
\end{theorem}

\begin{proof}
	Suppose $\cSet$ is finite. Let $N$ be the largest value appearing in any of the configurations 
	of $\cSet$. It is easy to see that if $h > |Q| \cdot N$, then any configuration of size
	$h$ does not belong in $\cSet$. It follows that if $h > |Q| \cdot N$ and $C \in \poststar((\multiset{h \cdot \init},\#))$ then $C \in \prestar(\uparrow \fin)$ and 
	so we have a positive cut-off.
	
	Suppose $\cSet$ is infinite. Then there must be an infinite set of configurations
	which belong to $\poststar(([\init],\#))$ but not in $\prestar(\uparrow \fin)$.
	This means that for infinitely many numbers $h$, there is a configuration $C_h \in \poststar((\multiset{h \cdot \init},\#))$ but $C_h \notin \prestar(\uparrow \fin)$.
	This then implies that $\ASMS$ cannot have a positive cut-off.
\end{proof}

Hence, checking whether $\ASMS$ has a positive cut-off is equivalent to
deciding if $\cSet$ is finite. We shall now show that this is decidable in \PSPACE{}.
To show this, we recall the connection established between RBN and ASMS in Section 4 of~\cite{Gandalf21}.

Given an ASMS $\ASMS =(Q,\Sigma,\delta)$, in section 4.2 of~\cite{Gandalf21}, it is shown that
in polynomial time we can come up with an RBN $\RBN = (Q \cup \Sigma \cup Q', \Sigma', \delta')$
which has a copy of $Q$ and $\Sigma$ as its states and which has the following properties:
\begin{itemize}
	\item A good configuration of $\RBN$ is a configuration $C$ such that $\sum_{a \in \Sigma} C(a) = 1$ and $C(q) = 0$ if $q \notin Q \cup \Sigma$. Notice that there is a natural bijection
	between configurations of $\ASMS$ and good configurations of $\RBN$. 
	\item A configuration $C$ of $\ASMS$ can reach a configuration $C'$ of $\ASMS$ iff 
	$\hat{C}$ can reach $\hat{C'}$ in $\RBN$, where $\hat{C}$ and $\hat{C'}$ are
	the corresponding good configurations of $C$ and $C'$ respectively.
\end{itemize}

Let $\initcube$ denote the set of all good configurations of $\RBN$ and let
$\cube$ denote the set of all good configurations of $\RBN$ which puts 
an arbitrary number of agents in $\init$, exactly one agent in $\#$ and zero agents elsewhere.

It then follows that
the set $\cSet := \poststar(([\init],\#)) \cap \overline{\prestar(\uparrow \fin)}$ over $\ASMS$ is finite iff 
the set $\cSet' := (\poststar(\cube) \cap \initcube) \cap \overline{\prestar(\uparrow \fin) \cap \initcube}$ over $\RBN$ is finite.
Similar to the proof of Lemma~\ref{lem:almost-sure-upper-bound}, we can decide if this set is finite
or not in \PSPACE{}. This gives the required \PSPACE{} upper bound.